\documentclass[a4paper, 11pt]{article}

\usepackage{graphicx}%
\usepackage{multirow}%
\usepackage{amsmath,amssymb,amsfonts}%
\usepackage{amsthm}%
\usepackage{mathrsfs}%
\usepackage[title]{appendix}%
\usepackage{xcolor}%
\usepackage{textcomp}%
\usepackage{manyfoot}%
\usepackage{booktabs}%
\usepackage{algorithm}%
\usepackage{algorithmicx}%
\usepackage{algpseudocode}%
\usepackage{listings}%
\usepackage{fullpage,url,here,tikz-cd}

\graphicspath{{Figs/}{./}}
	
\def\dtheta{\mathrm{d}\theta}
\def\Sym{\mathrm{Sym}}
\def\bbR{\mathbb{R}}
\def\bbL{\mathbb{L}}
\def\mattwotwo#1#2#3#4{\left[\begin{array}{cc}#1 & #2 \cr #3 & #4\end{array}\right]}
\def\inner#1#2{\langle #1,#2\rangle}
\def\KL{\mathrm{KL}}
\def\tr{\mathrm{tr}}
\def\dx{\mathrm{d}x}
\def\dy{\mathrm{d}y}
\def\dX{\mathrm{d}X}
\def\dY{\mathrm{d}Y}
\def\bbH{\mathbb{H}}
\def\bbC{\mathbb{C}}
\def\SL{\mathrm{SL}}
\def\SO{\mathrm{SO}}
\def\dz{\mathrm{d}z}
\def\dzbar{\mathrm{d}\bar{z}}
\def\Im{\mathrm{Im}}
\def\dt{\mathrm{d}t}
\def\SU{\mathrm{SU}}

\def\tr{\mathrm{tr}}
\def\calP{\mathcal{P}}
\def\MInner#1#2{{[#1,#2]}}

\def\TV{\mathrm{TV}}
\def\calX{\mathcal{X}}

\newtheorem{definition}{Definition}
\newtheorem{theorem}{Theorem}
\newtheorem{proposition}{Proposition}

\newtheorem{lemma}{Lemma}
\newtheorem{example}{Example}
\newtheorem{corollary}{Corollary}
\newtheorem{remark}{Remark}

\title{Information measures and geometry of the hyperbolic exponential families of Poincar\'e and hyperboloid  distributions} 

\author{
Frank Nielsen\\
Sony Computer Science Laboratories Inc.\\
E-mail: {\tt Frank.Nielsen@acm.org}\\
\and 
Kazuki Okamura\\
Department of Mathematics, Faculty of Science, Shizuoka University\\
E-mail:  {\tt okamura.kazuki@shizuoka.ac.jp}
}

\date{}

\begin{document}

\maketitle

\begin{abstract}
Hyperbolic geometry has become popular in machine learning due to its capacity to embed hierarchical graph structures with low distortions for further downstream processing. 
It has thus become important to consider statistical models and inference methods for data sets grounded in hyperbolic spaces.
In this paper, we study various information-theoretic measures and the information geometry of the Poincar\'e distributions and the related hyperboloid distributions, and prove that their statistical mixture models are universal density estimators of smooth densities in hyperbolic spaces.
The Poincar\'e and the hyperboloid distributions are two types of hyperbolic probability distributions defined using different models of hyperbolic geometry. 
Namely, the Poincar\'e distributions form a triparametric bivariate exponential family whose sample space is the hyperbolic Poincar\'e upper-half plane  and natural parameter space is the open 3D convex cone of  two-by-two positive-definite matrices.
The family of hyperboloid distributions form another exponential family which has sample space the forward sheet of the two-sheeted unit hyperboloid modeling hyperbolic geometry.

In the first part, we prove that all Ali-Silvey-Csisz\'ar's $f$-divergences between 
Poincar\'e distributions  can be expressed using three canonical terms using 
the framework of maximal group invariance. 
We also show that the $f$-divergences between any two Poincar\'e distributions are asymmetric except when those distributions belong to a same leaf of a particular foliation of the parameter space.
We report a closed-form formula for the Fisher information matrix, 
the Shannon's differential entropy and the Kullback-Leibler divergence  
between such distributions using the framework of exponential families.
In the second part, we state the corresponding results for the exponential family of hyperboloid distributions by highlighting a  parameter correspondence between the Poincar\'e and the hyperboloid distributions. 
Finally, we describe a random  generator to draw variates and present two Monte Carlo  methods to 
estimate numerically $f$-divergences between hyperbolic distributions. 
\end{abstract}


\tableofcontents

\section{Introduction}

\subsection{Statistical modeling in hyperbolic models}

Previous studies have shown that hyperbolic geometry\footnote{Hyperbolic geometry has constant negative curvature and the volume of hyperbolic balls increases exponentially with respect to their radii rather than polynomially as in Euclidean space.}~\cite{HG-2006} is very well suited for embedding tree graphs with low distortions~\cite{Sarkar-2011}  as hyperbolic Delaunay subgraphs of embedded tree nodes. 
Therefore, a recent trend in machine learning and data science is to embed discrete hierarchical graphs into continuous spaces with low distortions for further downstream tasks~\cite{Lorentz-2018,HyperbolicEmbeddings-2018,Ganea-2018,PoincareEmbedding-2017,suris2021learning,shimizu2021hyperbolic,song2022preliminary,montanaro2022rethinking,grover2022public,cho2022rotated}.
There exists many models of hyperbolic geometry~\cite{HG-2006} like the 
Poincar\'e disk or upper-half plane conformal models, the Klein non-conformal disk model, the Beltrami hemisphere model, the Minkowski  or Lorentz hyperboloid model, etc. 
We can transform one model of hyperbolic geometry to another model by a bijective mapping yielding a corresponding isometric embedding~\cite{cannon1997hyperbolic}. 
The hyperbolic plane can also be realized partially as a 2D surface in the 3D Euclidean space called the Beltrami pseudosphere~\cite{stillwell1996sources}.
Those various models of hyperbolic geometry (also refered in the literature as Lobachevskii space~\cite{andreev1970convex,troshin2017generalization} or Bolyai-Lobachevsky space~\cite{ungar2012mobius}) have pros and cons depending on the considered applications.
For example, Klein disk model is well-suited to compute hyperbolic Voronoi diagrams since the hyperbolic Voronoi bisectors are affine in Klein model~\cite{nielsen2010hyperbolic}. However, since Klein model is not conformal, the hyperbolic Voronoi diagram is often visualized in the Poincar\'e conformal disk model.
In the video\footnote{Available online at \url{https://www.youtube.com/watch?v=i9IUzNxeH4o}}~\cite{HVD-2014},
 the hyperbolic Voronoi diagrams is shown in the five main models of hyperbolic distribution.

\begin{figure}
\begin{center}
\includegraphics[width=0.8\columnwidth]{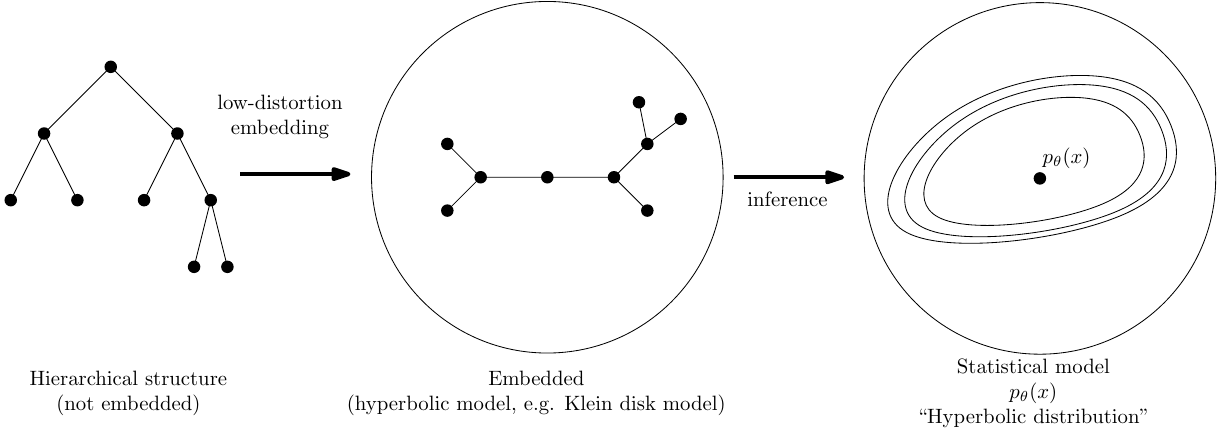}
\end{center}
\caption{Statistical modeling in a hyperbolic model: A hierarchical structure (left) is embedded in a hyperbolic model with low-distortion (middle). The point data set is then modeled by a probability distribution in the hyperbolic model (right).}
\label{fig:pipeline}
\end{figure}

As a byproduct of the low-distortion hyperbolic embeddings of hierarchical graphs, many embedded data sets are available in hyperbolic  model spaces, and those data sets need to be further processed.
Thus it is important to build statistical models and inference methods for these hyperbolic data sets using probability distributions with support hyperbolic  model spaces and to consider statistical mixtures in those spaces.
Figure~\ref{fig:pipeline} displays the pipeline to get probability distributions in hyperbolic geometry from hierarchical structures.
We quickly review the various families of probability distributions defined in hyperbolic models as follows:
\begin{itemize}
\item One of the very first proposed family of such ``hyperbolic distributions'' was proposed in 1981~\cite{HyperboloidDistribution-1981} and are nowadays commonly called the {\em hyperboloid distributions}.
The hyperboloid distributions are defined on the Minkowski upper sheet hyperboloid by analogy to the von-Mises Fisher distributions~\cite{vMFclustering-2005} which are defined on the sphere.

\item Barbaresco~\cite{GaussianPoincare-2019} defined the so-called Souriau-Gibbs distributions (2019) in the Poincar\'e disk (Eq. 57 of~\cite{GaussianPoincare-2019}, a natural exponential family) with its Fisher information metric coinciding with the Poincar\'e hyperbolic  Riemannian metric (the Poincar\'e unit disk is a homogeneous space with symmetry group $\SU(1,1)$). 

\item Nagano et al.~\cite{nagano2019wrapped} presented a pseudo-hyperbolic Gaussian (a hyperbolic wrapped normal distribution) whose density can be calculated analytically and differentiated with simple random variate generation algorithm (see also~\cite{cho2022rotated}).

\item Recently, Tojo and Yoshino~\cite{tojo-2020} (2020)  described a generic method~\cite{tojo2021harmonic} to build exponential families of distributions on homogeneous spaces which are invariant under the action of a Lie group, and illustrate their method with an example of an exponential family distribution supported on the Poincar\'e upper-half plane with natural parameter the cone of symmetric positive-definite matrices.
They also report its conjugate prior distributions, which are a piece of convenient mathematical machinery used in Bayesian statistics and Bayesian learning~\cite{ConjugatePrior-1979,agarwal2010geometric}. 
\end{itemize}

\subsection{Contributions and paper outline}

In information sciences, it is important to consider the dissimilarity between two different distributions in a family of distributions.
As an example, let us mention that dissimilarities between empirical distributions and statistical models allow to define
divergence-based estimators, which are robust to different types of noise compared to the maximum likelihood estimator (MLE) 
~\cite{eguchi2014duality,eguchi2022minimum}. 
In particular, bounded $f$-divergences such as the Hellinger divergence provide numerical stability~\cite{beran1977minimum} compared to MLE, which relies on the Kullback-Leibler divergence.  
The Kullback-Leibler divergence is unbounded and numerically unstable in practice.
$f$-divergences are also related to surrogate loss functions in machine learning and Bayes' risk in statistical decision theory~ \cite{nguyen2009surrogate}.

The notion of $f$-divergences between two distributions is one of the most principled methods to measure dissimilarity: $f$-divergences are invariant under sufficiency and can be easily symmetrized.
There are many possibilities for the generators $f$ beyond the specific one that yields the Kullback-Leiber divergence.
However, $f$-divergences are defined by integrals which can be  difficult to calculate in closed form or to evaluate numerically robustly. 
 
The main purpose of this paper is to investigate the $f$-divergences between Poincar\'e distributions and the $f$-divergences between hyperboloid distributions by using geometric approaches.
Our main result is that every $f$-divergence is a function of certain three terms of the parameters of each of the models by finding a maximal invariant of a group action.   
This structural result shows that all $f$-divergences can be defined by formula syntax trees with leaves labeled by these terms.
We corroborate our result by giving explicit formulae for several $f$-divergences.
Finally, we report the Riemannian metric tensors and the Amari-Chentsov cubic tensors which allow to derive the $\pm\alpha$-geometry in information geometry~\cite{IG-2016}, which includes the Riemannian Fisher-Rao geometry ($\alpha=0$).

We first consider in the first part various information-theoretic measures and geometry~\cite{IG-2016} of the family of Poincar\'e distributions in \S\ref{sec:Poincare}. 
Since the family of Poincar\'e distributions form an exponential family~\cite{EF-2014}, the underlying information-geometric structure of the family viewed as a smooth manifold is dually flat\footnote{Dually flat manifolds are also called Bregman manifods~\cite{nielsen2021geodesic,IG-2022} since they admit canonical Bregman divergences~\cite{Bregman-1967,IG-2016}.} ~\cite{IG-2016,shima2007geometry}.
We  prove using the 
method of group action maximal invariants~\cite{Eaton-1989} that all Ali-Silvey-Csisz\'ar's $f$-divergences~\cite{csiszar1964informationstheoretische,ali1966general} (including the Kullback-Leibler divergence) between Poincar\'e distributions can be expressed canonically as functions of three terms (Proposition~\ref{prop:maxinv} and Theorem~\ref{thm:fdivPoincare}).
We further show that on a particular foliation of the natural parameter space of the Poincar\'e distributions, the $f$-divergences are symmetric on the leaves of constant determinant\footnote{See also the foliation of the symmetric positive-definite space equipped with the affine-invariant trace metric~\cite[Section~4]{nielsen2012hyperbolic}. See also \cite{amari1997information,moakher2005differential} for general information for the foliation.} of the foliation. 
Then we report the Fisher information matrix of the Poincar\'e family (Eq.~\eqref{eq:fim-Poincare-2D}), 
the Amari-Chentsov cubic tensor for $\alpha$-geometry (with $0$-geometry corresponding to the Fisher-Rao geometry~\cite{IG-2016}) and various information-theoretic quantities like the differential entropy of a Poincar\'e distribution (Proposition~\ref{prop:h-Poincare} and Eq.~\eqref{eq:hPoincare}), or the Kullback-Leibler divergence (Proposition~\ref{prop:kld-Poincare} and Eq.~\eqref{eq:kldPoincare})  
between two Poincar\'e distributions.

In \S\ref{sec:hyperboloid}, we first define the hyperboloid distributions in arbitrary dimension and give their exponential family canonical decomposition. 
The Fisher information matrix of the Poincar\'e family (Eq.~\eqref{eq:fim-hyp-2D}) is reported. 
We also prove that mixtures of hyperboloid distributions are universal density approximators of smooth densities on the hyperboloid in \S\ref{sec:mixuda}.
A correspondence in \S\ref{sec:correspondence} between the upper-half plane and the Minkowski hyperboloid 2D sheet is exhibited. 
The $f$-divergences between the hyperboloid distributions are very geometric because we have a beautiful and clear maximal invariant, which has connections with the side-angle-side congruence criteria for triangles in hyperbolic geometry. 

In \S\ref{sec:numerics}, 
we provide two Monte Carlo approximation algorithms for calculating numerically the $f$-divergences.
Finally, we conclude in \S\ref{sec:concl}.

\section{The family of Poincar\'e distributions: An exponential family with Poincar\'e upper-half sample space}\label{sec:Poincare}
Tojo and Yoshino~\cite{tojo2021harmonic,tojo-2020} described a versatile method to build ``interesting'' exponential families of distributions on homogeneous spaces which are invariant under the action of a Lie group $G$ generalizing the construction in~\cite{cohen2015harmonic}.
They exemplify their method on the upper-half plane 
$$
\bbH:=\{(x,y)\in\bbR^2 :\ y>0\}
$$
by constructing  an exponential family with probability density functions invariant under the action of  Lie group $G=\SL(2,\bbR)$, the set of invertible matrices with unit determinant.
We shall call these distributions the Poincar\'e distributions since their sample space $\calX=\bbH$, and study this set of distributions as an exponential family~\cite{EF-2014}.

\subsection{An exponential family parameterized either by vectors or matrices}\label{sec:PoincareVM}

In this subsection, we first give the definition of the Poincar\'e distribution, and then, give two canonical forms of the density function of the Poincar\'e distribution as an  exponential family by using vectors or matrices as the parameters.

\renewcommand{\arraystretch}{2.5}
\begin{figure}
\centering
\begin{tabular}{|l|c|c|c|}\hline
$\bbH$ & \includegraphics[width=0.3\textwidth]{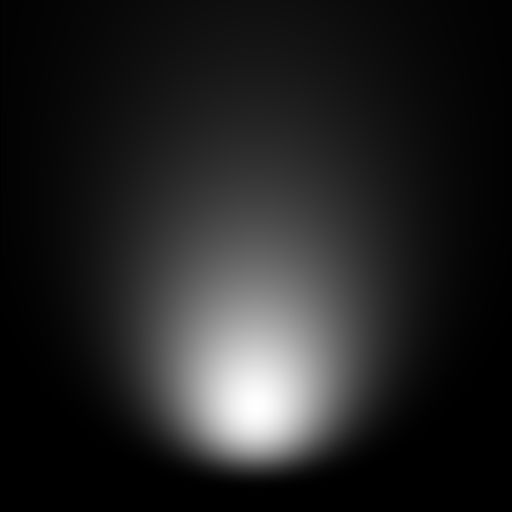}&
\includegraphics[width=0.3\textwidth]{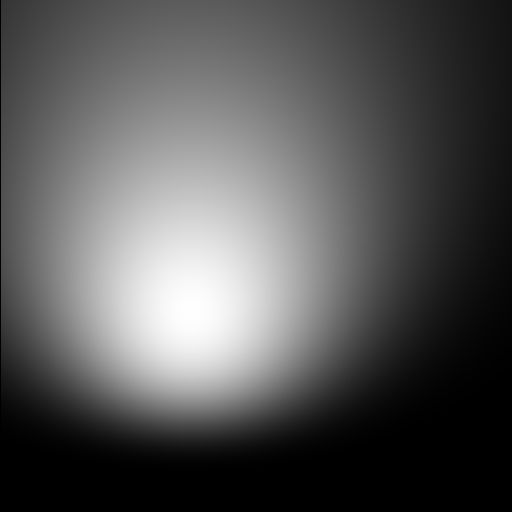}&
\includegraphics[width=0.3\textwidth]{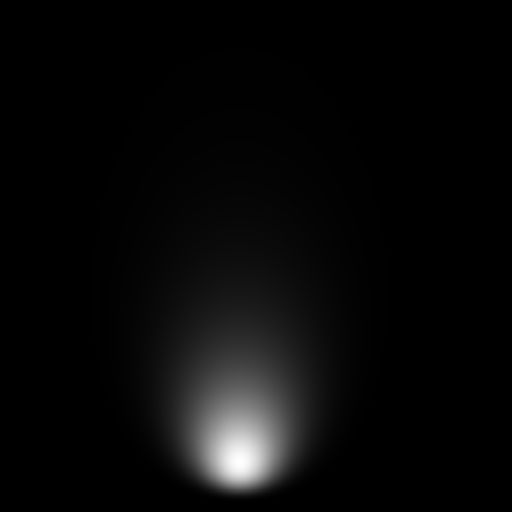}\\
$\theta\in\Theta$ & $\theta=\mattwotwo{1}{0}{0}{1}$ &
$\theta=\mattwotwo{\frac{1}{2}}{\frac{1}{4}}{\frac{1}{4}}{2}$ &
$\theta=\mattwotwo{2}{\frac{1}{4}}{\frac{1}{4}}{\frac{1}{2}}$\\
$\theta_v$ & $\theta_v=(1,0,0)$ & $\theta_v=\left(\frac{1}{2},\frac{1}{4},2\right)$
&  $\theta_v=\left(2,\frac{1}{4},\frac{1}{2}\right)$\\ \hline
\end{tabular}

\caption{Plotting the probability density function $p_\theta(x)$  of the Poincar\'e distribution on the Poincar\'e upper-half plane indexed by a $2\times 2$ positive-definite matrix $\theta$ (plotted for $x\in [-2,2]\times (0,2]\subset\bbH$). \label{fig:density}}
\end{figure}
\renewcommand{\arraystretch}{1.5}

The probability density function (pdf) of a Poincar\'e distribution~\cite{tojo-2020} expressed using a 3D vector parameter $\theta_v=(a,b,c)\in\bbR^3$ is given by
\begin{equation*}
p_{\theta_v}(x,y) :=  \frac{\sqrt{ac-b^2}\exp(2\sqrt{ac-b^2})}{\pi}\, \exp\left(- \frac{a(x^2+y^2)+2bx+c}{y}\right)\,\frac{1}{y^2},
\end{equation*}
where  $\theta_v$ belongs to the parameter space 
\begin{equation*}
\Theta_v:=\{(a,b,c)\in\bbR^3\ : a>0, c>0, \ ac-b^2>0\}.
\end{equation*}
The set $\theta_v$ forms an open 3D convex cone.
Thus the Poincar\'e distribution family has a 3D parameter cone space and the sample space is the hyperbolic upper plane. 
Figure~\ref{fig:density} displays three examples of density profiles of Poincar\'e distributions.

In general, an exponential family~\cite{EF-2014,nielsen2009statistical} $\calP=\{p_\theta(x)\ :\ \theta\in\Theta, \Theta \subset \bbR^n \}$ 
has its probability density functions which can be written canonically as
\begin{equation}\label{eq:canonical-exp-fam}
p_\theta(x)=\exp\left(\theta^\top t(x)-F(\theta)+k(x)\right),
\end{equation}
where $t(x)$ denotes the sufficient statistic vector, $\theta$ the natural parameter, $\theta^\top$ the transpose of $\theta$, $k(x)$ an auxiliary carrier term and $F(\theta)$ the log-normalizer (also called free energy or log-partition in statistical physics and cumulant function in statistics). 
 Let $h(x)=e^{k(x)}$ so that $p_\theta(x)=\exp\left(\theta^\top t(x)-F(\theta)\right)\, h(x)$.

Let $z=x+iy\in\bbC$ and consider the mapping 
\begin{equation*}
\alpha(z):=\mattwotwo{\sqrt{y}}{\frac{x}{\sqrt{y}}}{0}{\frac{1}{\sqrt{y}}}.
\end{equation*}
Then the Poincar\'e pdf rewrites as
\begin{equation*}
p_{\theta}(z) = \frac{\sqrt{|\theta|}e^{2\sqrt{|\theta|}}}{\pi} 
\exp\left( -\tr\left(\theta \alpha(z)^\top \alpha(z)\right) \right)\, \frac{1}{\Im(z)^2},  
\end{equation*}
where $\theta=\mattwotwo{a}{b}{b}{c}$ belongs to the space of symmetric positive-definite matrices $\Sym^+(2,\bbR)$, $|\theta|=ac-b^2$ denotes the determinant of $\theta$,  and $\tr(\cdot)$ denotes the matrix trace (e.g., $\tr(\theta)=a+c$). 

The family $\mathcal{P}=\{p_\theta\ :\ \theta\in \Sym^+(2,\bbR)\}$  is a bivariate exponential family of order $3$ with
log-normalizer which can be expressed either using the vector parameter $\theta_v$ or the positive-definite matrix $\theta$:
\begin{eqnarray*}
F(\theta_v) &=& \log\left(\frac{\pi}{\sqrt{ac-b^2}\exp(2\sqrt{ac-b^2})}\right),\\
F(\theta ) &=& \log\left(\frac{\pi}{\sqrt{|\theta|}\exp(2\sqrt{|\theta|})}\right),
\end{eqnarray*}

It will be useful in the remainder to define the square root of the determinant of $\theta$:
\begin{equation*}
D(\theta_v) =\sqrt{|\theta|}=\sqrt{ac-b^2}>0, \ \theta_v = (a,b,c). 
\end{equation*}
We have $F(\theta_v)=\left(\log \frac{\pi}{D}\right)-2D$.

A Bregman generator $F: \theta_v \in \bbR^{3} \rightarrow \bbR$ is a strictly convex $C^3$ real-valued function. 
A Bregman generator allows to build a dually flat space (i.e., a Hessian manifold with a single global chart) equipped with a canonical Bregman divergence:  
$$
B_F(\theta_v:\theta_{v}^{\prime})=F(\theta_v)-F(\theta_{v}^{\prime})-(\theta_v-\theta_{v}^{\prime})^\top \nabla F(\theta_{v}^{\prime}).
$$

Let $G(\theta_v) := F(\theta_v)+ w^{\top}\theta_v + \lambda $ for $w \in \bbR^3$  and $\lambda\in\bbR$.  
Then we have $B_F(\theta_v:\theta_{v}^{\prime})=B_G(\theta_v:\theta_{v}^{\prime})$.
That is, the Bregman generators $F$ and $G$ are equivalent modulo affine terms $w^{\top}\theta_v +\lambda$ ~\cite{banerjee2005clustering}.

Thus, by deleting $\pi$, we may consider 
$$
F(\theta_v) \equiv -\log D - 2D = -\frac{1}{2}\log |\theta| - \sqrt{|\theta|}. 
$$
It is well-known that $-\log |\theta|$ is a strictly convex function~\cite{dhillon2008matrix}.

The sufficient statistic vector $t_v(x,y)$ is 
$$
t_v(x,y)=-\left(\frac{x^2+y^2}{y},\frac{x}{y},\frac{1}{y}\right),  
$$
or in equivalent matrix form:
$$
t(z)=-\alpha(z)\alpha(z)^\top = -\frac{1}{y} \mattwotwo{x^2+y^2}{x}{x}{1}.
$$

The auxiliary carrier term is $h(x,y)=e^{k(x,y)}=\frac{1}{y^2}$, and $h(x,y)\dx\dy$ is a $\SL(2,\bbR)$-invariant measure, or equivalently 
$h(z)=e^{k(z)}=\frac{1}{y^2}$ and $h(z)\dz\dzbar$ is a $\SL(2,\bbR)$-invariant measure.

Thus we can write these Poincar\'e densities in the following vector/matrix canonical forms of exponential families~\cite{nielsen2009statistical}:
\begin{eqnarray}
p_{\theta_v}(x,y) &=&\exp\left(\inner{\theta_v}{t(x,y)}-F(\theta_v)+k(x,y)\right),\\
p_{\theta}(z) &=& \exp\left(\inner{\theta}{t(z)}-F(\theta)+k(z)\right),
\end{eqnarray}
where $\inner{v_1}{v_2}=v_1^\top v_2$ denotes the vector inner product (dot product) and $\inner{M_1}{M_2}$ denotes the matrix inner product $\inner{M_1}{M_2}=\tr(M_1 M_2^\top)$.
Table~\ref{tab:Pparam} summarizes the dual vector/matrix parameterizations of the Poincar\'e distributions.
The Poincar\'e distributions are related to the hyperboloid distributions~\cite{HyperboloidDistribution-1981} as mentioned in~\cite{tojo-2020,tojo-2022}.

\begin{table}
\centering
\scalebox{0.8}[0.8]{
\begin{tabular}{|l|l|l|}\hline
 & Vector form & Matrix form\\ \hline
Inner product & $\inner{v_1}{v_2}=v_1^\top v_2$&  $\inner{M_1}{M_2}=\tr(M_1 M_2^\top)$\\
Probability density & $\exp\left(\inner{\theta_v}{t(x,y)}-F(\theta_v)\right)h(x,y)$ & $\exp\left(\inner{\theta}{t(z)}-F(\theta)\right)h(z)$\\
Cumulant function $F(\theta)$ & $\log\left(\frac{\pi}{\sqrt{ac-b^2}\exp(2\sqrt{ac-b^2})}\right)$ & $\log\left(\frac{\pi}{\sqrt{|\theta|}\exp(2\sqrt{|\theta|})}\right)$\\
Natural parameter space & $\{(a,b,c)\in\bbR^3\ : a>0, c>0, \ ab-c^2>0\}$ & $\Sym^+(2,\bbR)$\\
Sufficient statistic $t(x,y)$ &  $-\left(\frac{x^2+y^2}{y},\frac{x}{y},\frac{1}{y}\right)$ & $-\frac{1}{y} \mattwotwo{x^2+y^2}{x}{x}{1}$\\ \hline
Carrier measure $h(x,y)$ & \multicolumn{2}{|c|}{$\frac{1}{y^2}\dx\dy=\frac{1}{y^2}\dz\dzbar$}\\  \hline  
Moment parameter  & $-\frac{1}{ac-b^2}\left(\frac{1}{2}+\sqrt{ac - b^2}\right) (c,-2b,a)^\top$ & $-\left(\frac{1}{2}+\sqrt{|\theta|}\right) \theta^{-\top}$\\ 
$\eta=E[t(x,y)]=\nabla F(\theta)$ & & \\ \hline
\end{tabular}
}
\caption{Dual vector/matrix parameterizations of the Poincar\'e distributions defined on the upper plane sample space $\bbH=\{z=x+iy\in\bbC\ :\ y>0\}$.\label{tab:Pparam}}
\end{table}

\subsection{Statistical $f$-divergences between Poincar\'e distributions}\label{sec:fdiv}

In this subsection, we show that every $f$-divergence between Poincar\'e distributions is a function of a triplet consisting of certain functions defined by the parameters and give explicit formulae for the squared Hellinger divergence and the Neyman chi-squared divergence.  

The $f$-divergence~\cite{csiszar1964informationstheoretische,ali1966general} induced by a convex generator $f:\bbR_{++}\rightarrow \bbR$ between two pdfs $p(x,y)$ and $q(x,y)$ defined on the  support $\bbH$ is defined by
$$
D_f[p:q] := \int_{(x,y) \in \bbH} p(x,y)\, f\left(\frac{q(x,y)}{p(x,y)}\right)\, \dx\,\dy. 
$$ 
Since $D_f[p:q]\geq f(1)$, we consider convex generators $f(u)$ such that $f(1)=0$. 
The class of $f$-divergences include the total variation distance ($f(u)=|u-1|$), 
the Kullback-Leibler divergence ($f(u)=-\log(u)$, and its two common symmetrizations, namely, the Jeffreys divergence and the Jensen-Shannon divergence), the squared Hellinger divergence, the Pearson and Neyman sided $\chi^2$-divergences, etc.

We state the notion of maximal invariant by following~\cite{Eaton-1989}:
Let $G$ be a group acting on a set $X$. 
We denote it by $(g,x) \mapsto gx$.

\begin{definition}[Maximal invariant of a group action~\cite{Eaton-1989}]
We say that a map $\varphi$ from $X$ to a set $Y$ is {\it maximal invariant} if it is invariant, namely,  
$\varphi(gx) = \varphi(x)$ for every $g \in G$ and $x \in X$, and furthermore, whenever $\varphi(x_1) = \varphi(x_2)$ there exists $g \in G$ such that $x_2 = gx_1$. 
Every invariant map is a function of a maximal invariant. 
Specifically, 
if a map $\psi$ from $X$ to a set $Z$ is invariant, 
then, there exists a unique map $\Phi$ from $\varphi(X)$ to $Z$ such that $\Phi \circ \varphi = \psi$. The spaces and functions are summarized in this commutative diagram: 

\begin{center}
\begin{tikzcd}
X \arrow[rd,"\psi"] \arrow[r, "\varphi"] & Y \arrow[d,"\Phi"] \\
& Z
\end{tikzcd}
\end{center}

\end{definition}

For each $x \in X$, 
we may consider its orbit $O_x:=\{gx \in X \ :\ g\in G\}$. 
A map is invariant when it is constant on orbits and maximal invariant when orbits have distinct map values. 
We denote by $A^{\top}$  
the transpose of a square matrix $A$ and $A^{-\top}$  the transpose of the inverse matrix $A^{-1}$ of a regular matrix $A$. 
It holds that $A^{-\top} = (A^{\top})^{-1}$. 
 
The action of $g=\mattwotwo{a}{b}{c}{d}\in\SL(2,\bbR)$ with $ad-bc=1$ (for $a,b,c,d\in\bbR$) on the sample space $z\in\bbH$ is a linear fractional transformation:
\begin{equation*}
z\mapsto \frac{az+b}{cz+d}.
\end{equation*}

Furthermore, the action $g.z=\frac{az+b}{cz+d}$ corresponds to the action of $g$ on $\theta \in \Sym^{+}(2, \mathbb R)$: 
\begin{equation*}
g.\theta:=g^{-\top}\times \theta\times g^{-1}.
\end{equation*}

For $g \in \SL(2,\bbR)$, 
we denote the pushforward measure of a measure $\nu$ on $\mathbb{H}$ by the map $z \mapsto g.z$ on $\mathbb{H}$ by $\nu \circ g^{-1}$. 
Then, by the proof of \cite[Proposition 1]{tojo-2020}, $p_{\theta} \circ g^{-1} = p_{g.\theta}$.  

Since $f$-divergence is invariant under the pushforward by a diffeomorphism, we obtain that 
\begin{proposition}\label{prop:inv-Poincare}
$$D_{f} \left[p_{\theta}:p_{\theta^{\prime}} \right] = D_{f} \left[p_{\theta} \circ g^{-1} :p_{\theta^{\prime}}  \circ g^{-1}\right] = D_f \left[p_{g^{-\top} \theta g^{-1}}:p_{g^{-\top} \theta^{\prime} g^{-1}} \right].$$
\end{proposition}

Therefore, it is natural to consider the following action of $\SL(2,\mathbb R)$ on the product of the parameter space $\Sym^{+}(2, \mathbb R)^2$. 
\begin{proposition}[Maximal invariant]\label{prop:maxinv}
Define a group action of $\SL(2, \mathbb R)$ on $\Sym^{+}(2, \mathbb R)^2$ by 
\begin{equation*}
\left(g, (\theta, \theta^{\prime})\right) \mapsto (g^{-\top} \theta g^{-1}, g^{-\top} \theta^{\prime} g^{-1}).  
\end{equation*}
Define a map $S: \Sym^+(2,\bbR)^2 \to (\bbR_{>0})^2 \times \mathbb{R}$ by 
\begin{equation*}
S(\theta, \theta^{\prime}) := \left(|\theta|, |\theta^{\prime}|, \textup{tr}(\theta^{\prime} \theta^{-1})\right).
\end{equation*}
Then, the map $S$ is maximal invariant of the group action. 
\end{proposition}

\begin{proof}
Observe that $S$ is invariant with respect to the group action:
$$
S\left(\theta, \theta^{\prime}\right) = S\left(g.\theta, g.\theta^{\prime}\right).
$$

Assume that $S\left(\theta^{(1)}, \theta^{(2)}\right) = S\left(\widetilde{\theta^{(1)}}, \widetilde{\theta^{(2)}}\right)$. 
We see that there exists $g_{\theta^{(1)}} \in \SL(2, \mathbb R)$ such that $g_{\theta^{(1)}}. \theta^{(1)} = g_{\theta^{(1)}}^{-\top} \theta^{(1)} g_{\theta^{(1)}}^{-1} = \sqrt{|\theta^{(1)}|} I_2$, where $I_2$ denotes the $2\times 2$ identity matrix.  
Then, $\theta^{(1)} =  \sqrt{|\theta^{(1)}|} g_{\theta^{(1)}}^{\top} g_{\theta^{(1)}}$. 
Let $\theta^{(3)} := g_{\theta^{(1)}}. \theta^{(2)}  = g_{\theta^{(1)}}^{-\top} \theta^{(2)} g_{\theta^{(1)}}^{-1}$.  
Then, 
\[ \textup{tr}\left(\theta^{(3)}\right) = \textup{tr}\left(\theta^{(2)} g_{\theta^{(1)}}^{-1} g_{\theta^{(1)}}^{-\top}\right) = \sqrt{|\theta^{(1)}|}\, \textup{tr}\left(\theta^{(2)} (\theta^{(1)})^{-1}\right). \]

We define $g_{\widetilde{\theta^{(1)}}}$ and $\widetilde{\theta^{(3)}}$ in the same manner. 
Then, $\textup{tr}\left(\theta^{(3)}\right) = \textup{tr}\left(\widetilde{\theta^{(3)}}\right)$ and $|\theta^{(3)}| = \left|\widetilde{\theta^{(3)}}\right|$. 
Hence the set of eigenvalues of $\theta^{(3)}$ and $\widetilde{\theta^{(3)}}$ are identical with each other. 
By this and $\theta^{(3)}, \widetilde{\theta^{(3)}} \in \textup{Sym}(2, \mathbb R)$, 
there exists $h \in \SO(2)$ such that $h.\theta^{(3)} = \widetilde{\theta^{(3)}}$.  
Hence $(hg_{\theta^{(1)}}). \theta^{(2)} = g_{\widetilde{\theta^{(1)}}} \widetilde{\theta^{(2)}}$. 
We also see that $$(hg_{\theta^{(1)}}). \theta^{(1)} = g_{\theta^{(1)}}. \theta^{(1)} = \sqrt{\left|\theta^{(1)}\right|}\, I_2 = \sqrt{\left|\widetilde{\theta^{(1)}}\right|}\, I_2 = g_{\widetilde{\theta^{(1)}}}. \widetilde{\theta^{(1)}}. $$
Thus we have 
$$
\left(\widetilde{\theta^{(1)}}, \widetilde{\theta^{(2)}}\right) = (g_{\theta^{(1)}}^{-1}hg_{\theta^{(1)}}).(\theta^{(1)}, \theta^{(2)}).
$$
\end{proof}

See \cite[Remark 1]{nielsen2023i} for an alternative proof of Proposition \ref{prop:maxinv}.  

By Propositions \ref{prop:inv-Poincare} and \ref{prop:maxinv}, 
\begin{theorem}[Canonical terms of the $f$-divergences between Poincar\'e distributions]\label{thm:fdivPoincare}
Every $f$-divergence between two Poincar\'e distributions $p_{\theta}$ and $p_{\theta^{\prime}}$ is a function of $ \left(|\theta|, |\theta^{\prime}|, \textup{tr}\left(\theta^{\prime} \theta^{-1}\right)\right)$ 
and invariant with respect to the $\SL(2,\bbR)$-action. 
\end{theorem}

Since $|\theta+\theta^{\prime}| = |\theta| + |\theta^{\prime}| + |\theta| \textup{tr}\left(\theta^{\prime} \theta^{-1}\right)$, 
every $f$-divergence between two Poincar\'e distributions $p_{\theta}$ and $p_{\theta^{\prime}}$ is also a function of $ \left(|\theta|, |\theta^{\prime}|, |\theta+\theta^{\prime}|\right)$.

Recently, Tojo and Yoshino \cite{tojo2023} introduced a notion of deformed exponential family associated with their $G/H$ method in representation theory. 
As an example of it,  they considered a family of {\it deformed Poincar\'e distributions}. 
The statement of Theorem \ref{thm:fdivPoincare} above also holds for these distributions. 
See \cite[Theorem 2]{nielsen2023i} for details. 

\begin{remark}[Foliation]\label{rmk:symfdiv}
For $t > 0$, let $\Theta(t) := \{\theta \in \Theta : |\theta| = t\}$.
The sets $\Theta(t)$ for $t>0$ yields a foliation of the natural parameter space $\Theta$: $\Theta=\cup_{t>0} \Theta(t)$.
Then, 
every $f$-divergence between distributions on $\Theta(t)$ is symmetric:
$D_f(p_{\theta_1}:p_{\theta_2})=D_f(p_{\theta_2}:p_{\theta_1})$ for all $\theta_1,\theta_2\in\Theta(t)$. 
\end{remark}

We have exact formulae for the squared Hellinger divergence and the Neyman chi-squared divergence.  
\begin{proposition}\label{prop:SHNC-Poincare}
Let $\theta, \theta^{\prime} \in \Theta$. Then,\\
(i) (squared Hellinger divergence) 
Let $f(u) = (\sqrt{u} - 1)^2 / 2$. 
Then, 
\begin{equation*}
D_f [p_{\theta}: p_{\theta^{\prime}}] 
= 1 - \frac{2|\theta|^{1/4}|\theta^{\prime}|^{1/4} 
\exp\left( |\theta|^{1/2} + |\theta^{\prime}|^{1/2} \right)}{ |\theta+\theta^{\prime}|^{1/2}  \exp\left(|\theta+\theta^{\prime}|^{1/2}\right)}.
\end{equation*}
(ii) (Neyman chi-squared divergence) 
Let $f(u) := (u-1)^2$. 
Then, 
\[ D_f [p_{\theta}: p_{\theta^{\prime}}] 
= \begin{cases} \frac{|\theta^{\prime}| \exp(4|\theta^{\prime}|^{1/2})}{|\theta|^{1/2} |2\theta^{\prime}-\theta|^{1/2} \exp\left(2 (|\theta|^{1/2} + |2\theta^{\prime}-\theta|^{1/2}) \right)} - 1 & 2\theta^{\prime} - \theta \in \Theta \\ +\infty & \textup{ otherwise} \end{cases} \]
\end{proposition}

We remark that $|\theta+\theta^{\prime}|$ and $|2\theta^{\prime}-\theta|$ can be expressed by using $|\theta|, |\theta^{\prime}|$, and $\textup{tr}(\theta^{\prime}\theta^{-1})$. 
Indeed, we have
\begin{eqnarray*}
|\theta+\theta^{\prime}| &=& |\theta| + |\theta^{\prime}| + |\theta|\, \textup{tr}(\theta^{\prime}\theta^{-1}),\\
|2\theta^{\prime}-\theta| &=& 4|\theta^{\prime}| + |\theta| -2|\theta|\,\textup{tr}(\theta^{\prime}\theta^{-1}).
\end{eqnarray*}

We can show Proposition \ref{prop:SHNC-Poincare} by straightforward calculations. 
The Kullback-Leibler divergence is considered in the following subsection. 

\begin{proof}
Let $\theta=\mattwotwo{a}{b}{b}{c} \in \Theta$ and $\theta^{\prime}=\mattwotwo{a^{\prime}}{b^{\prime}}{b^{\prime}}{c^{\prime}} \in \Theta$.  

(i) We remark that $\theta + \theta^{\prime} \in \Theta$.  
Then, 
\[ \sqrt{p_{\theta}(x,y) p_{\theta^{\prime}}(x,y)} \]
\[= \frac{|\theta|^{1/4} |\theta^{\prime}|^{1/4}  \exp\left(\left(|\theta|^{1/2} + |\theta^{\prime}|^{1/2}\right)/2\right)}{\pi y^2} \exp\left(-\frac{(a+a^{\prime})(x^2+y^2)/2 + (b+b^{\prime})x + (c+c^{\prime})/2}{y}\right). \]
It holds that 
\[ \int_{\bbH} \exp\left(-\frac{(a+a^{\prime})(x^2+y^2)/2 + (b+b^{\prime})x + (c+c^{\prime})/2}{y}\right) \frac{\dx\dy}{y^2} \]
\[= \frac{2\pi}{\sqrt{(a+a^{\prime})(c+c^{\prime}) - (b+b^{\prime})^2}\exp(\sqrt{(a+a^{\prime})(c+c^{\prime}) - (b+b^{\prime})^2}/2)}.\]
Now the assertion holds.

(ii) 
Assume $2\theta^{\prime} - \theta \in \Theta$. 
It holds that 
\[ \frac{p_{\theta^{\prime}}(x,y)^2}{p_{\theta}(x,y)} = \frac{|\theta^{\prime}|\exp(4|\theta^{\prime}|^{1/2})}{\pi |\theta|^{1/2}\exp(2|\theta|^{1/2})} 
\exp\left(-\frac{(2a^{\prime}-a)(x^2+y^2) + 2(2b^{\prime}-b)x + 2c^{\prime}-c}{y}\right).\]
By the assumption, 
it holds that 
\[ \int_{\bbH} \exp\left(-\frac{(2a^{\prime}-a)(x^2+y^2) + 2(2b^{\prime}-b)x + 2c^{\prime}-c}{y}\right) \frac{\dx\dy}{y^2} \]
\[= \frac{\pi}{\sqrt{(2a^{\prime}-a)(2c^{\prime}-c) - (2b^{\prime}-b)^2}\exp(\sqrt{(2a^{\prime}-a)(2c^{\prime}-c) - (2b^{\prime}-b)^2})}.\]

Assume $2\theta^{\prime} - \theta \notin \Theta$. 
Then, the assertion follows from Corollary \ref{cor:exfdiv-2d} (iii) and the correspondence principle in Section \ref{sec:correspondence} below. 
\end{proof}

\begin{example}
Let us report a numerical example. 
Let $\theta_1=\mattwotwo{1}{0}{0}{1}$ and $\theta_2=\mattwotwo{\frac{1}{2}}{0}{0}{2}$.
Both $p_{\theta_1}$ and $p_{\theta_2}$ belongs to the left $\Theta(1)$ since $|\theta_1|=|\theta_2|=1$.
Then $D_f(p_{\theta_1}:p_{\theta_2})=D_f(p_{\theta_2}:p_{\theta_1})=\frac{3}{4}$.
\end{example}

A family of conjugate priors of a density 
$p_\theta(x)=\exp(\inner{\theta}{t(x)}-F(\theta))h(x)$ of a $m$-order exponential family is an exponential family of order $m+1$ with sample space $\calX_c=\Theta$
 and canonical density
 $$
q_\vartheta(\theta)=\exp(\inner{\vartheta}{(\theta,-F(\theta))}-F_c(\vartheta))h_c(\theta). 
 $$ 
 Conjugate priors are used in Bayesian inference, and the conjugate prior of the Poincar\'e distribution is an exponential family of order $4$ which has been reported in~\cite{tojo-2020}. 
 See also \cite{ConjugatePrior-1979} for some background on conjugate priors.

\subsection{Kullback-Leibler divergences from reverse Bregman divergences}\label{sec:kld}

In this subsection, we give explicit formulae for the Kullback-Leibler divergence between Poincar\'e distributions by using the Bregman generator. 

The {\it Kullback-Leibler divergence} $D_{\KL}$ is the $f$-divergence obtained for the generator $f(u)=- \log u$. 
Let $D(\theta_1:\theta_2)$ be a divergence.
The reverse, backward or dual divergence $D^*$ of $D$ is defined by just swapping the parameter arguments as follows:  
$$
D^*(\theta:\theta^{\prime}) := D(\theta^{\prime}:\theta).
$$
In general, $D^* \ne D$ (\cite{azoury2001relative}). 
Since the Kullback-Leibler divergence between two densities of an exponential family amounts to a reverse Bregman divergence,  
we have
$$
D_{\KL}[p_{\theta}:p_{\theta^{\prime}}]=\int_{x=-\infty}^\infty\int_{y=0}^\infty p_{\theta}(x,y)\log \frac{p_{\theta}(x,y)}{p_{\theta^{\prime}}(x,y)} \dx\dy=B_F(\theta^{\prime}:\theta),
$$
where 
$$
F(\theta)\equiv -\frac{1}{2}\log |\theta| -2\sqrt{|\theta|}
$$ 
since Bregman generators are equivalent modulo affine terms.  
The log-normalizer can be expressed as a function of $D$: $F(\theta)=-\log D-2D + \log \pi$ where $D=D(\theta)=\sqrt{|\theta|} =D(a,b,c)=\sqrt{ac-b^2}$. 

\begin{figure}
\centering
\includegraphics[width=0.65\textwidth]{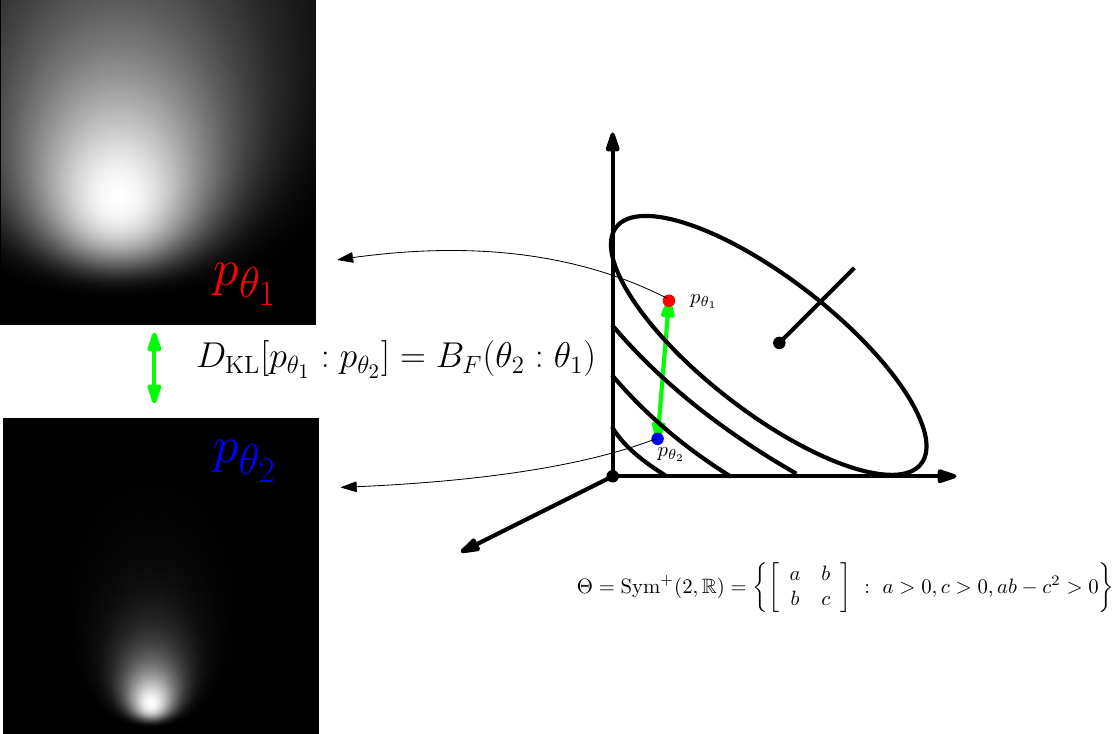}

\caption{Correspondence between calculating the Kullback-Leibler divergence between parametric densities   and the corresponding  parameter divergence on the cone parameter space}\label{fig:spddist}
\end{figure}

Figure~\ref{fig:spddist} schematically illustrates the correspondence between calculating a statistical divergence between parametric densities (e.g., the KLD) and the corresponding  parameter divergence on the parameter space (e.g., the natural parameter space $\Theta$ is the open convex symmetric positive-definite cone, SPD cone for short).

The matrix gradient of $F(\theta)$ is
\begin{equation*}
\nabla F(\theta)=-\left(\frac{1}{2}+\sqrt{|\theta|}\right) \theta^{-\top} =-\frac{1}{ac-b^2}\left(\frac{1}{2}+\sqrt{ac - b^2}\right) \, (c,-2b,a)^\top, \  \ \theta = (a,b,c)^\top,  
\end{equation*}
where $\theta^{-\top}=(\theta^{-1})^\top$ denotes the inverse transpose operator.

Thus we have 
\begin{eqnarray*}\label{eq:kldPoincare}
D_{\KL}[p_{\theta}:p_{\theta^{\prime}}]&=& B_F(\theta^{\prime}:\theta) := F(\theta^{\prime})-F(\theta)-\inner{\theta^{\prime}-\theta}{\nabla F(\theta)},\\
&=& \frac{1}{2}\log\frac{|\theta|}{|\theta^{\prime}|}+2\left(\sqrt{|\theta|}-\sqrt{|\theta^{\prime}|}\right)
+\left(\frac{1}{2}+\sqrt{|\theta|}\right)(\tr(\theta^{\prime}\theta^{-1})-2). 
\end{eqnarray*}

\begin{proposition}\label{prop:kld-Poincare}
The Kullback-Leibler divergence between two Poincar\'e distributions $p_{\theta}$ and $p_{\theta^{\prime}}$ is
\begin{equation*}
D_\KL[p_{\theta}:p_{\theta^{\prime}}] = \frac{1}{2}\log\frac{|\theta|}{|\theta^{\prime}|}+2\left(\sqrt{|\theta|}-\sqrt{|\theta^{\prime}|}\right)
+\left(\frac{1}{2}+\sqrt{|\theta|}\right)(\tr(\theta^{\prime}\theta^{-1})-2).
\end{equation*}
\end{proposition}

Observe that the KLD is indeed a function of $D=|\theta|$, $D^{\prime}=|\theta^{\prime}|$ and $\tr(\theta^{\prime}\theta^{-1})$ as claimed in Corollary~\ref{thm:fdivPoincare}.

By Proposition \ref{prop:inv-Poincare}, it holds that for any $g\in\SL(2,\bbR)$, 
\begin{equation*}
D_\KL[p_{\theta}:p_{\theta^{\prime}}]=D_\KL[g.p_{\theta}:g.p_{\theta^{\prime}}]=D_\KL[p_{g.\theta}:p_{g.\theta^{\prime}}].
\end{equation*}

The Mahalanobis divergence between two vectors $v$ and $v^{\prime}$ induced by a symmetric positive-definite matrix $Q$ is
$$
\Delta_Q^2(v:v^{\prime})=(v-v^{\prime})^\top\, Q\, (v-v^{\prime}).
$$ 

The Mahalanobis distance \cite{mahalanobis1936generalised} $\Delta^2(N_1,N_2)$ between two normal distributions $N_1=N(\mu_1,\Sigma)$ and $N_2=N(\mu_2,\Sigma)$ is defined by 
$$
\Delta^2(N_1,N_2)=\Delta^2_{\Sigma^{-1}}(\mu_1,\mu_2).
$$

Notice that the only symmetric Bregman divergences are squared Mahalanobis divergences~\cite{BVD-2010}.

Thus the KLD between two Poincar\'e distributions is asymmetric in general. 
The situation is completely different from the Cauchy distribution whose $f$-divergences are always symmetric~\cite{nielsen2022f, Verdu2023}.

\subsection{Dual moment parameterization and Fenchel-Young divergences}\label{sec:fyd}

In this subsection, we discuss the information-geometric duality of Poincar\'e distributions, namely the canonical dual natural and moment parameterizations.    
The dual moment parameterization~\cite{EF-2014} $\eta$ of an exponential family distribution $p_\theta$ with log-normalizer $F(\theta)$  is given by $\eta=\nabla F(\theta)$.
The dual parameter $\eta$ is also called the moment parameter because the following property holds: $E_{p_\theta}[t(x)]=\nabla F(\theta)=\eta$.

Observe that $\eta$ is a symmetric negative-definite matrix, and therefore the moment parameter space is
$$
H= \left\{ \eta(\theta)=\nabla F(\theta)\ :\ \theta\in\Theta=\Sym^+(2,\bbR) \right\} = \Sym^{-}(2,\bbR).
$$

In general, the convex conjugate function $F^*(\eta)$ of $F(\theta)$ is defined by the Legendre-Fenchel transform:
$$
F^*(\eta) = \sup_{\theta\in\Theta} \inner{\theta}{\eta}-F(\theta).
$$
The supremum is unique because $\nabla_\theta^2 \left(\inner{\theta}{\eta}-F(\theta)\right)=-\nabla^2_\theta F(\theta)$ is negative-definite since $F(\theta)$ is strictly convex,
 and can be obtained by solving the equation $\nabla_\theta \left(\inner{\theta}{\eta}-F(\theta)\right)=0$, that is, $\nabla F(\theta)=\eta$.
Since $F$ is a smooth strictly convex function of Legendre type, there exists a unique global inverse function $(\nabla F)^{-1}$, and 
we have $\theta(\eta)=(\nabla F)^{-1}(\eta)$. 
The convex conjugate function is then given by $F^*(\eta)=\inner{\eta}{\theta(\eta)}-F(\theta(\eta))$, and we have the reciprocal gradients of convex conjugate functions: $\theta(\eta)=\nabla F^*(\eta)$ and $\eta(\theta)=\nabla F(\theta)$.

Consider the vector natural parameterization $\theta_v$ of Poincar\'e distributions.
The problem is to calculate $\theta_v(\eta_v)$ which amounts to compute the reciprocal gradient function $(\nabla F)^{-1}$
since we can then express $D(\theta_v)$ as $D(\eta_v)=D(\theta_v(\eta_v))$, and we shall get the convex conjugate function as
$$
F^*(\eta_v)=\log D(\theta_v (\eta_v))-1-\log\pi. 
$$

 The gradient with respect to $\theta_v=(a,b,c)$ is: 
\begin{eqnarray*}
\eta_v &=&
\left[
\begin{array}{l}
A\\
B\\
C\\
\end{array}
\right]
=\nabla F(\theta_v)=
\left[
\begin{array}{c}
-E_{p_{\theta}}\left[\frac{x^2+y^2}{y}\right]\\
-E_{p_{\theta}}\left[\frac{x}{y}\right]\\
-E_{p_{\theta}}\left[\frac{1}{y}\right]
\end{array}
\right],\\
&=&-\frac{1}{D^2(\theta_v)}\left(\frac{1}{2}+D(\theta_v)\right) \,
\left[
\begin{array}{l}
c\\
-2b\\
a
\end{array}
\right]
.
\end{eqnarray*}
Thus we need to express the natural parameter $\theta_v=(a,b,c)$ as a function of the moment parameter $\eta_v=(A,B,C)$.

\begin{proposition}[Dual moment parameterization]
Assume that $\eta_v = \nabla F (\theta_v)$ and let 
$\eta_v = \left[\begin{array}{c}
A\\
B\\
C\\
\end{array}
\right]$. 
Then, we have
$$
 \theta_v(\eta_v) =
\frac{1}{2(E(\eta_v)-E^2(\eta_v))}
\left[\begin{array}{c}
C\\
-B/2\\
A\\
\end{array}\right]= \left[\begin{array}{c}
a\\
b\\
c\\
\end{array}
\right]=\nabla F^*(\eta_v),
$$
where $E(\eta_v)=\sqrt{AC-B^2/4}$.
\end{proposition}

\begin{proof}
Let $D := D(\theta_v)$. 
Let $\alpha := \frac{1}{D^2} \left(\frac{1}{2} + D\right)$. 
Since $D > 0$, $\alpha > 0$. 
We see that 
\begin{equation}\label{eq:simple-expression-theta} 
\theta_v = -\frac{1}{\alpha} \left[\begin{array}{c}
C\\
-B/2\\
A\\
\end{array}
\right]. 
\end{equation}
Hence, 
$$
 D = \frac{\sqrt{4AC-B^2}}{2\alpha}.
$$
Therefore we have: 
$$
\alpha = \frac{1}{D^2} \left(\frac{1}{2} + D\right) = \frac{2\alpha (\alpha + \sqrt{4AC-B^2})}{4AC-B^2}.
$$
By this and $\alpha > 0$, 
$$
\alpha = \frac{4AC-B^2}{2} - \sqrt{4AC-B^2}.
$$
By this and Eq.~\eqref{eq:simple-expression-theta}, we obtain the assertion. 
\end{proof}

We have $\inner{\eta_v}{\theta_v(\eta_v)}=\frac{E(\eta_v)}{1-E(\eta_v)}$
and $\inner{\theta_v}{\eta_v(\theta_v)}=-1-2 D(\theta_v)$.

It follows that the convex conjugate function is given in closed-form by 
$F^*(\eta_v)=\log D(\theta(\eta_v))-1-\log\pi$.
(This convex conjugate function shall be used to ease the calculation of the differential entropy of Poincar\'e distributions as reported in~\S\ref{sec:h}.)

The dual natural and moment parameterizations can be used to express the Kullback-Leibler divergence between two Poincar\'e distributions  $p_{\theta}$ and $p_{\theta^{\prime}}$  into four different equivalent forms as follows:

\begin{eqnarray*}
D_\KL[p_{\theta}:p_{\theta^{\prime}}] &=& B_F(\theta^{\prime}:\theta)=B_{F^*}(\eta:\eta^{\prime}),\\
&=& A_{F,F^*}(\theta^{\prime}:\eta)= A_{F^*,F}(\eta,\theta^\prime),
\end{eqnarray*}
where $A_{F,F^*}(\theta,\eta')$ is the Fenchel-Young divergence~\cite{nielsen2021geodesic}, which expresses the Bregman divergence $B_F(\theta:\theta')$ using the mixed  natural and moment parameterizations: 
\begin{equation*}
A_{F,F^*}(\theta:\eta^{\prime})=F(\theta)+F^*(\eta^{\prime})-\inner{\theta}{\eta^{\prime}}
=B_F(\theta:\theta^\prime).
\end{equation*}
This is a consequence of the Fenchel-Moreau biconjugation theorem in \cite{rockafellar1970convex}, which holds for proper lower semicontinuous convex functions $F$. 

In information geometry, the natural coordinate system $\theta$ and the dual moment coordinate system $\eta$ of an exponential family are mutually orthogonal~\cite{IG-2016}. 
This can be attested by the Crouzeix identity~\cite{crouzeix1977relationship} of the Hessians of the conjugate functions: $\nabla^2 F(\theta)\nabla^2 F^*(\eta)=I$, where $I$ denotes the identity matrix.

Exponential families enjoy many nice properties~\cite{EF-2014}:
For example, the maximum likelihood estimator (MLE) of
 a set of $n$ identically and independently distributed observations $z_1, \ldots, z_n$ on $\bbH$ is given in the moment parameterization as follows: 
$$
\hat\eta=\frac{1}{n}\sum_{i=1}^n t(z_i).
$$
Distribution $p_{\theta(\hat\eta)}$ is called the observed point in information geometry~\cite{IG-2016}. 
The MLE is consistent, asymptotically normally distributed, and efficient~\cite{EF-2014} (i.e., matching the Cram\'er-Rao lower bound).
The corresponding natural parameter is $\hat{\theta}=\theta(\hat\eta)=\nabla F^*(\hat\eta)$ by the equivariance property of the MLE.

 Since both $\eta_v(\theta_v)=\nabla F(\theta_v)$ and $\theta_v(\eta_v)=\nabla F^*(\eta_v)$ are available in closed-form, we can express the exponential geodesics and mixture geodesics in closed-form~\cite{IG-2016} ($e$-geodesics and $m$-geodesics, respectively) for the Poincar\'e family.
The $e$-geodesic $\gamma(p_{\theta_1},p_{\theta_2};\lambda)$ passing through $p_{\theta_1}$ for $\lambda=0$ and $p_{\theta_2}$ for $\lambda=1$ is expressed as a linear interpolation in the natural parameterization:
$$
\theta(\gamma(p_{\theta_1},p_{\theta_2};\lambda))=(1-\lambda)\theta_1+\lambda\theta_2,\lambda\in [0,1].
$$
Similarly, the dual $m$-geodesic 
$\gamma^*(p_{\eta_1},p_{\eta_2};\lambda)$ passing through $p_{\eta_1}$ for $\lambda=0$ and $p_{\eta_2}$ for $\lambda=1$ is expressed as a linear interpolation in the moment parameterization:
$$
\eta(\gamma^*(p_{\eta_1},p_{\eta_2};\lambda))=(1-\lambda)\eta_1+\lambda\eta_2,\lambda\in [0,1].
$$
The dual $m$-geodesic can be expressed in the natural parameterization as follows:
$$
\theta(\gamma^*(p_{\theta_1},p_{\theta_2};\lambda))=
\theta(\eta(\gamma^*(p_{\theta(\eta_1)},p_{\theta(\eta_2)};\lambda)))=
\nabla F^{-1}\left((1-\lambda)\nabla F(\theta_1)+\lambda\nabla F(\theta_2)\right),\lambda\in [0,1].
$$
These dual $e$-geodesics and $m$-geodesics are often used by algorithms implemented on exponential family manifolds.
For example, we can approximate the smallest enclosing Kullback-Leibler ball of a set of Poincar\'e distributions using the algorithm in~\cite{nock2005fitting}.

\subsection{Fisher metric and $\alpha$-connection on Poincar\'e distributions}\label{sec:FIM} 

In this subsection, we discuss the Riemannian metric defined by the Fisher information matrix. 
Our purpose is to compute the Fisher information matrix (FIM) and the Amari-Chentsov cubic tensor explicitly.  

The Fisher information matrix for a regular exponential family~\cite{EF-2014} is defined as
\begin{equation*}
I_\theta(\theta) = -E_{p_\theta}\left[\nabla^2_\theta \log p_\theta(x)\right].
\end{equation*}

Thus for an exponential family of order $m=\dim(\Theta)$, we have $I_\theta(\theta)=\nabla^2 F(\theta)$.
Using the dual parameterization, we also have $I_\eta(\eta)=\nabla^2 F^*(\eta)$ and $I_\theta(\theta)I_\eta(\eta)=I_m$, where $I_m$ denotes the $m\times m$ identity matrix.

By using symbolic computing, 
we obtain the Fisher information matrix of the Poincar\'e distributions. 
In information geometry, the Fisher information matrix is used to define the Fisher information metric $g_F$ (expressed as the Fisher information matrix (FIM) $I_{\theta_v}$ in the local coordinate system $\theta_v$) on the Riemannian manifold of the Poincar\'e distributions $\calP$. The length element $\mathrm{d}s^2=\mathrm{d}\theta_v^\top \, I_{\theta_v}(\theta_v)
\, =\mathrm{d}\theta_v$ is independent of the parameterization.
The Rao distance~\cite{atkinson1981rao} is the geodesic distance of the Riemannian manifold $(\calP,g_F)$.

More generally, information geometry~\cite{IG-2016} considers the dual $\pm\alpha$-structures that consist of a  pair of torsion-free affine connections $\nabla^{\pm\alpha}$ coupled to the Fisher metric $g_F$ so that $\frac{\nabla^{\alpha}+\nabla^{-\alpha}}{2}=\nabla^g$, where $\nabla^g$ denotes the Levi-Civita metric connection induced by $g_F$. 
The Fisher-Rao geometry corresponds to the $0$-geometry.

The Fisher information matrix is given by 
\renewcommand\arraystretch{2}
\begin{equation}\label{eq:fim-Poincare-2D}
I(a,b,c)=
\begin{bmatrix}
{{c^2\,\sqrt{a\,c-b^2}+a\,c^3-b^2\,c^2}\over{\sqrt{a\,c-b^2}\,
 \left(2\,a^2\,c^2-4\,a\,b^2\,c+2\,b^4\right)}}&
 -{{b\,c\,\sqrt{a\,c-b^2}+a\,b\,c^2-b^3\,c}\over{\sqrt{a\,c-b^2}\,
 \left(a^2\,c^2-2\,a\,b^2\,c+b^4\right)}} &
-{{\left(a\,c-2\,b^2\right)\,\sqrt{a\,c-b^2}-b^2}\over{2\,a^2\,c^2-
 4\,a\,b^2\,c+2\,b^4}}\\
 -{{b\,c\,\sqrt{a\,c-b^2}+a\,b\,c^2-b^3\,c}\over{\sqrt{a\,c-b^2}\,
 \left(a^2\,c^2-2\,a\,b^2\,c+b^4\right)}} &
 {{2\,a\,c\,\sqrt{a\,c-b^2}+a\,c+b^2}\over{a^2\,c^2-2\,a\,b^2\,c+b^4
 }}&
 -{{a\,b\,\sqrt{a\,c-b^2}+a^2\,b\,c-a\,b^3}\over{\sqrt{a\,c-b^2}\,
 \left(a^2\,c^2-2\,a\,b^2\,c+b^4\right)}}\\
 -{{\left(a\,c-2\,b^2\right)\,\sqrt{a\,c-b^2}-b^2}\over{2\,a^2\,c^2-
 4\,a\,b^2\,c+2\,b^4}}&
-{{a\,b\,\sqrt{a\,c-b^2}+a^2\,b\,c-a\,b^3}\over{\sqrt{a\,c-b^2}\,
 \left(a^2\,c^2-2\,a\,b^2\,c+b^4\right)}}&
{{a^2\,\sqrt{a\,c-b^2}+a^3\,c-a^2\,b^2}\over{\sqrt{a\,c-b^2}\,
 \left(2\,a^2\,c^2-4\,a\,b^2\,c+2\,b^4\right)}}
\end{bmatrix}.
\end{equation}

We also have
$$
D_f[p_{a,b,c}:p_{a+\mathrm{d}a,b+\mathrm{d}b,c+\mathrm{d}c}] = \frac{f^{\prime\prime}(1)}{2}
 [\mathrm{d}a\ \mathrm{d}b\ \mathrm{d}c]\, I(a,b,c)\, \left[\begin{array}{l}\mathrm{d}a\\ \mathrm{d}b\\ \mathrm{d}c\end{array}\right].
$$

Notice that the Hessians of the conjugate functions  $\nabla^2 F(\theta)$ and $\nabla^2 F^*(\eta)$  correspond to the Fisher information matrices expressed in the natural $\theta$- and moment 
 $\eta$-parameterizations, respectively.

The $\alpha$-connections are defined by their Christoffel symbols $\Gamma^\alpha$ which can be found by the Christoffel symbols of the Levi-Civita connection and the Amari-Chentsov cubic tensor $T$ whose components are 
$T_{ijk}(\theta)=\frac{\partial}{\partial\theta_i}\frac{\partial}{\partial\theta_j}\frac{\partial}{\partial\theta_k} F(\theta)$.  
Using symbolic computing, 
we can calculate the $3^3=27$ components $T_{ijk}$ of $T$.
Since $T$ is totally symmetric, we have $T_{ijk}=T_{\sigma(i)\sigma(j)\sigma(k)}$ for any permutation $\sigma$ of $\{1,2,3\}$.
For example, we find that  
\begin{eqnarray*}
T_{111} &=& {{3\,c^3\,\sqrt{a\,c-b^2}-4\,c^3}\over{4\,a^3\,c^3-12\,a^2\,b^2\,c^
 2+12\,a\,b^4\,c-4\,b^6}},\\
 T_{222} &=&
{{6\,a^2\,b\,c^2+\sqrt{a\,c-b^2}\,\left(6\,a\,b\,c+2\,b^3\right)-6
 \,a\,b^3\,c}\over{\sqrt{a\,c-b^2}\,\left(a^3\,c^3-3\,a^2\,b^2\,c^2+3
 \,a\,b^4\,c-b^6\right)}}\\
 T_{333} &=&
-{{3\,a^3\,\sqrt{a\,c-b^2}+4\,a^3}\over{4\,a^3\,c^3-12\,a^2\,b^2\,c
 ^2+12\,a\,b^4\,c-4\,b^6}}\\
 T_{112} &=&
{{4\,b\,c^2\,\sqrt{a\,c-b^2}+3\,a\,b\,c^3-3\,b^3\,c^2}\over{\sqrt{a
 \,c-b^2}\,\left(2\,a^3\,c^3-6\,a^2\,b^2\,c^2+6\,a\,b^4\,c-2\,b^6
 \right)}}\\
T_{221} &=&
-{{a^2\,c^3+\sqrt{a\,c-b^2}\,\left(a\,c^2+3\,b^2\,c\right)+a\,b^2\,
 c^2-2\,b^4\,c}\over{\sqrt{a\,c-b^2}\,\left(a^3\,c^3-3\,a^2\,b^2\,c^2
 +3\,a\,b^4\,c-b^6\right)}}\\
T_{113} &=&
{{\sqrt{a\,c-b^2}\,\left(a\,c^2-4\,b^2\,c\right)-4\,b^2\,c}\over{4
 \,a^3\,c^3-12\,a^2\,b^2\,c^2+12\,a\,b^4\,c-4\,b^6}}\\
T_{331} &=&
{{\sqrt{a\,c-b^2}\,\left(a^2\,c-4\,a\,b^2\right)-4\,a\,b^2}\over{4
 \,a^3\,c^3-12\,a^2\,b^2\,c^2+12\,a\,b^4\,c-4\,b^6}}\\
T_{223} &=&
-{{a^3\,c^2+\sqrt{a\,c-b^2}\,\left(a^2\,c+3\,a\,b^2\right)+a^2\,b^2
 \,c-2\,a\,b^4}\over{\sqrt{a\,c-b^2}\,\left(a^3\,c^3-3\,a^2\,b^2\,c^2
 +3\,a\,b^4\,c-b^6\right)}}\\
T_{332} &=&
{{4\,a^2\,b\,\sqrt{a\,c-b^2}+3\,a^3\,b\,c-3\,a^2\,b^3}\over{\sqrt{a
 \,c-b^2}\,\left(2\,a^3\,c^3-6\,a^2\,b^2\,c^2+6\,a\,b^4\,c-2\,b^6
 \right)}}\\
T_{123} &=& {{-a^2\,b\,c^2+\sqrt{a\,c-b^2}\,\left(2\,a\,b\,c+2\,b^3\right)-a\,b
 ^3\,c+2\,b^5}\over{\sqrt{a\,c-b^2}\,\left(2\,a^3\,c^3-6\,a^2\,b^2\,c
 ^2+6\,a\,b^4\,c-2\,b^6\right)}}.
\end{eqnarray*}
See Appendix~A of arXiv 2205.13984 for the {\tt Maxima} source code of $T_{123}$. 
 
Note that this Fisher metric is invariant under the linear fractional transformation action of $\SL(2,\bbR)$.
Moreover, any $f$-divergence yields a scaled metric $f^{\prime\prime}(1) \, I_\theta(\theta)$ since
\begin{equation*}
I_g[p_{\theta}:p_{\theta+\dtheta}]=\frac{1}{2} f^{\prime\prime}(1) \dtheta^\top I_\theta(\theta) \dtheta + o\left(\|\dtheta\|^2\right),
\end{equation*}
when $\dtheta\rightarrow 0$, see~\cite{IG-2016}. 
In general, the Christoffel symbols $\bar{\Gamma}_{jk}^i$ of the metric Levi-Civita connection can be derived from the Fisher metric.  

Many Riemannian metrics and related Riemannian distances have been investigated and classified by invariance and other properties on the open cone of positive-definite matrices. 
We refer the reader to the PhD thesis~\cite{Thanwerdas-2022} for a panorama of such Riemannian metrics.

\subsection{Differential entropy}\label{sec:h}

In this subsection, we give an explicit formula for the differential entropy of a single Poincar\'e distribution and give an example.  

Recall that $k(x,y)$ is the auxiliary carrier term in \eqref{eq:canonical-exp-fam} and $e^{k(x,y)}=\frac{1}{y^2}$. 
The differential entropy $h[p_\theta]:=-\int p_\theta(x,y)\log p_\theta(x,y)\dx\dy $ can be calculated as follows~\cite{EntropyEF-2010}: 
\begin{equation*}
h[p_{\theta}]=-F^*(\eta)-E_{p_{\theta}}[k(x,y)]=-F^{*}(\eta)+2\,E_{p_{\theta}}[\log y].
\end{equation*}
Thus we need to calculate the term $E_{p_{\theta}}[\log y]$. 
Let $D := \sqrt{ac-b^2}$. 
We first integrate $\exp\left(- \frac{a(x^2+y^2)+2bx+c}{y}\right)$ with respect to $x$ and eliminate the variable $x$. 
Then, by the change of variable $y=\frac{D}{a} e^z$, 
\begin{equation*} 
E_{p_{\theta}}[\log y] = \sqrt{\frac{D}{\pi}} e^{2D} \int_{\mathbb R} \left(\log \frac{D}{a} + z\right) \exp\left( -2D \cosh(z) - \frac{z}{2} \right)  \dz. 
\end{equation*}
Let 
$K_{\nu}(z)$ be the modified Bessel function of second kind of order $\nu$ (~\cite{Gradshteyn2015}).  
Then, it has the following integral expression \cite[Eq.~(8.432.1)]{Gradshteyn2015}: 
\begin{equation*}
K_{\nu}(z) = \int_{0}^{\infty} \exp(-z \cosh(t)) \cosh(\nu t) dt, \ \ \textup{Re}(z) > 0. 
\end{equation*}
Then, 
\begin{eqnarray*}
\int_{\mathbb R} z \exp\left( -2D \cosh(z) - \frac{z}{2} \right) \dz &=& -2 \int_0^{\infty} z \exp(-2D \cosh(z)) \sinh\left(\frac{z}{2}\right) \dz,\\
 &=& -2 \frac{\partial}{\partial \nu}\biggm\vert_{\nu = 1/2} K_{\nu}(2D). 
\end{eqnarray*}
Therefore,
\begin{equation*}
E_{p_{\theta}}[\log y] = 2 \sqrt{\frac{D}{\pi}} e^{2D} \left(\log \left(\frac{D}{a}\right) K_{1/2}(2D) - \frac{\partial}{\partial \nu}\biggm\vert_{\nu = 1/2} K_{\nu}(2D) \right).
\end{equation*}

Since the equalities $K_{\frac{1}{2}}(x) = \sqrt{\frac{\pi}{2x}} e^{-x}$ and $\left.\frac{\partial}{\partial \nu}\right|_{\nu=\frac{1}{2}} K_{\nu}(x) = \sqrt{\frac{\pi}{2x}} e^{x} \Gamma(0,2x)$ holds for $x \in \bbR_{>0}$ (see \cite[Eq.~(8.469.3) and (8.486(1).21)]{Gradshteyn2015}, respectively, for example),
we obtain 
\begin{equation*}
E_{p_{\theta}}[\log y] = \log\left(\frac{D}{a}\right)-e^{4D}\,\Gamma(0,4D),
\end{equation*}
where
$\Gamma(a,x)=\int_x^\infty t^{a-1}e^{-t}\dt$ 
is the upper incomplete Gamma function.

Thus we have that 
\begin{proposition}\label{prop:h-Poincare}
The differential entropy $h[p_\theta]$ of a  Poincar\'e distribution $p_\theta$ is
\begin{equation}\label{eq:hPoincare}
h[p_\theta] = 1+\log(\pi D)  - 2\log a -2e^{4D}\Gamma(0,4D),
\end{equation}
where $D=\sqrt{|\theta|}=\sqrt{ac-b^2}$.
\end{proposition}

Notice that $h[p_\theta]$ is not a function of $D$. 
It depends also on $a$.  
See Remark \ref{rmk:mde} below for the reason. 

\begin{example}\label{ex:poincare}
Let us report a numerical example.
We consider two Poincar\'e distributions defined by the following parameters:
$$
\theta=\mattwotwo{4}{\frac{1}{4}}{\frac{1}{4}}{\frac{1}{2}},\quad
\theta^{\prime}=\mattwotwo{4}{\frac{1}{2}}{\frac{1}{4}}{\frac{1}{2}}.
$$

We have the dual parameters:
$$
\eta\simeq\mattwotwo{-0.488}{0.244}{0.244}{-3.906},\quad
\eta^{\prime}\simeq\mattwotwo{-3.132}{0.391}{0.391}{-0.783}.
$$

The dual potential functions are evaluated as
$$
F(\theta)\simeq -3.114, \quad
F(\theta^{\prime})\simeq -1.904,
$$
and
$$
F^*(\eta)\simeq -0.669, \quad
F^*(\eta^{\prime})\simeq -1.032,
$$

We find that the forward and reverse Kullback-Leibler divergences are 
$$
D_\KL[p_{\theta}:p_{\theta^{\prime}}]\simeq 5.360,\quad
D_\KL[p_{\theta^{\prime}}:p_{\theta}]\simeq 8.573 
$$

The differential entropies of $p_{\theta}$ and $p_{\theta^{\prime}}$ are
$$
h[p_{\theta}]\simeq -0.608,\quad
h[p_{\theta^{\prime}}]\simeq 3.074
$$

Now choose the following transformation matrix $g$ of $\SL(2,\bbR)$:
$$
g=\mattwotwo{1}{1}{1}{2}.
$$
Then we have
$$
g.\theta=\mattwotwo{15.5}{-7.75}{-7.75}{4},\quad
g.\theta^{\prime}=\mattwotwo{3}{-2.25}{-2.25}{2}.
$$

We see that the invariance of the KLD by the action of $g$:
$$
D_\KL[g.p_{\theta}:g.p_{\theta^{\prime}}]=D_\KL[p_{g.\theta}:p_{g.\theta^{\prime}}]=D_\KL[p_{\theta}:p_{\theta^{\prime}}]\simeq 5.360.
$$

Figure~\ref{fig:exinvariance} displays the distributions $p_{\theta}$, $p_{\theta^{\prime}}$,
 $p_{g.\theta}$, and $p_{g.\theta^{\prime}}$.
\end{example}

\begin{figure}\label{fig:exinvariance}
\centering
\begin{tabular}{cc}
$p_{\theta}$ & $p_{\theta^{\prime}}$\\
\includegraphics[width=0.24\textwidth]{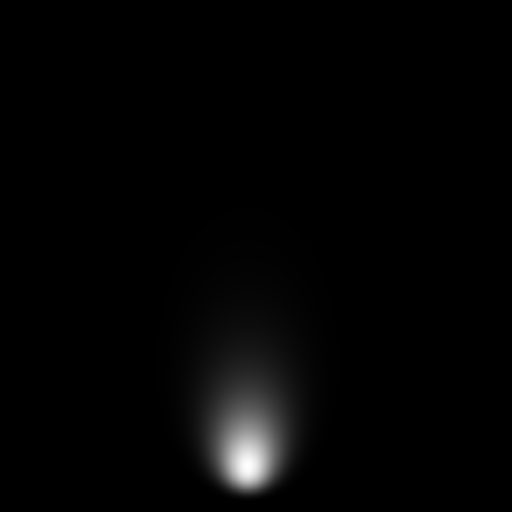} &
\includegraphics[width=0.24\textwidth]{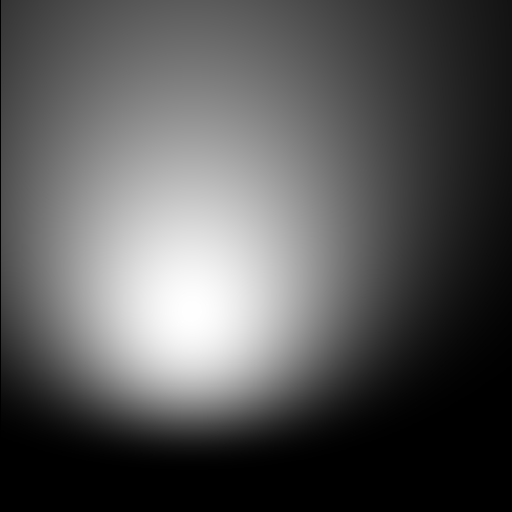}\\
$p_{g.\theta}$ & $p_{g.\theta^{\prime}}$\\
\includegraphics[width=0.24\textwidth]{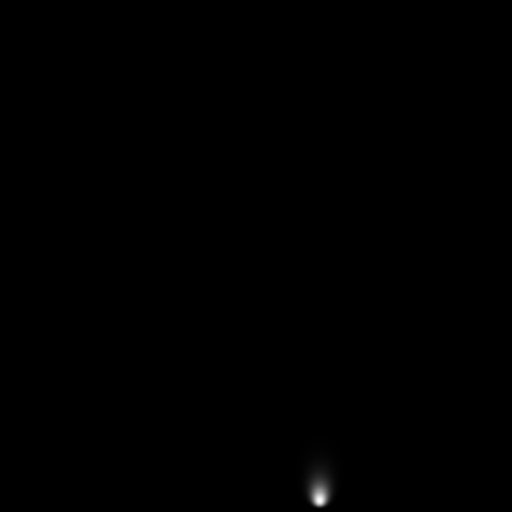} &
\includegraphics[width=0.24\textwidth]{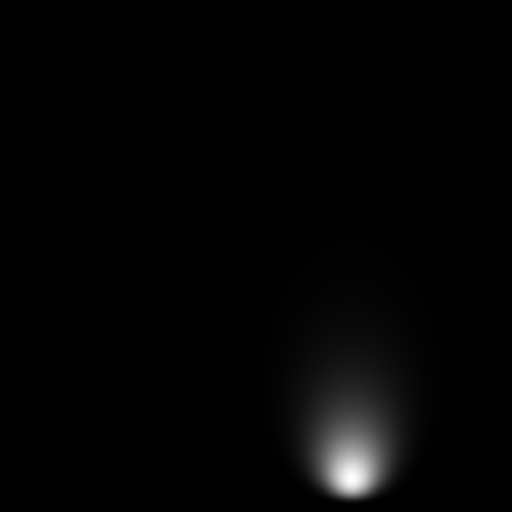}
\end{tabular}
\caption{The Poincar\'e distributions $p_{\theta}$ and $p_{\theta^{\prime}}$ with their corresponding distributions $p_{g.\theta}$and $p_{g.\theta^{\prime}}$ obtained by the action of $g\in\SL(2,\bbR)$. 
The Kullback-Leibler divergence is preserved: $D_\KL[p_{g.\theta}:p_{g.\theta^{\prime}}]=D_\KL[p_{\theta}:p_{\theta^{\prime}}]$}
\end{figure}

\subsection{The skew Bhattacharyya distances and Chernoff information}\label{sec:skewBhat}

The {\it $\alpha$-Bhattacharyya divergence}~\cite{chernoff1952measure,BR-2011} between two probability distributions with densities $p(x,y)$ and $q(x,y)$ on the support $\bbH$ is defined by
\begin{equation}
D_\alpha[p:q] := -\log \int_{\bbH} p^\alpha(x,y) q^{1-\alpha}(x,y) \dx\dy.
\end{equation}

When the densities belong the same exponential family with cumulant function $F(\theta)$, i.e.,
$p=p_{\theta}$ and $q=p_{\theta^{\prime}}$ of an exponential family,
the $\alpha$-Bhattacharyya divergence amounts to a skew Jensen divergence:
\begin{equation}
D_\alpha[p_{\theta}:p_{\theta^{\prime}}]=J_{F,\alpha}(\theta:\theta^{\prime}),
\end{equation}
where
\begin{equation}
J_{F,\alpha}(\theta:\theta^{\prime})=\alpha F(\theta)+(1-\alpha)F(\theta^{\prime})-F(\alpha\theta+(1-\alpha)\theta^{\prime}).
\end{equation}

Moreover, the KLD between densities of an exponential family  tends asymptotically to a scaled  skewed Jensen divergence~\cite{BR-2011}:
\begin{equation}
D_\KL[p_{\theta}:p_{\theta^{\prime}}] =\lim_{\alpha\rightarrow 0} \frac{1}{\alpha(1-\alpha)}  J_{F,\alpha}(\theta:\theta^{\prime})
\end{equation}
with the skewed Jensen divergence for the Poincar\'e family is
\begin{eqnarray}\label{eq:BhatPoincare}
J_{F,\alpha}(\theta:\theta^{\prime})
&=& \frac{1}{2}\log\left( \frac{|(1-\alpha)\theta+\alpha\theta^{\prime}|}{|\theta|^{1-\alpha}\, |\theta^{\prime}|^\alpha}\right)\nonumber\\
&&
+2\left(\sqrt{|(1-\alpha)\theta+\alpha\theta^{\prime}|}-((1-\alpha)\sqrt{|\theta|}+\alpha\sqrt{|\theta^{\prime}|})\right).\label{eq:skewjd}
\end{eqnarray}

By choosing $\alpha$ small (say, $\alpha=\epsilon=0.01$), we can approximate the KLD by a scaled $\alpha$-skewed Jensen divergence which does not require to calculate the gradient term $\nabla F(\theta)$:
\begin{equation}
D_\KL[p_{\theta}:p_{\theta^{\prime}}] \approx_{\epsilon} J_{F,\epsilon}(\theta:\theta^{\prime}).
\end{equation}

The Chernoff information between two Poincar\'e distributions can be approximated efficiently using~\cite{nielsen2013information}.

\section{The hyperboloid distributions}\label{sec:hyperboloid}  

We first give the definition of the Lobachevskii space (in reference to Minkowski hyperboloid model of hyperbolic geometry also called the Lorentz model) and the parameter space of the hyperboloid distribution. 
Let $d \ge 2$. 
Let 
$$
\mathbb{L}^d := \left\{(x_0,x_1,\cdots,x_d) \in \mathbb{R}^{d+1}  : x_0 = \sqrt{1 + x_1^2 + \cdots + x_d^2} \right\} 
$$
and 
$$\Theta_{\mathbb{L}^d} := \left\{(\theta_0,\theta_1,\cdots,\theta_d) \in \mathbb{R}^{d+1} : \theta_0 > \sqrt{\theta_1^2 + \cdots +\theta_d^2} \right\}.
$$

\begin{figure}
\centering
\includegraphics[width=0.8\columnwidth]{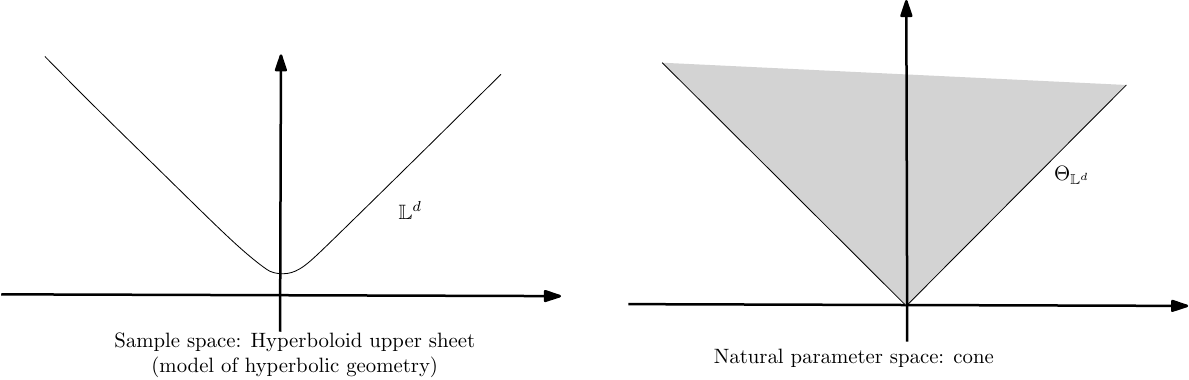}
\caption{The family of hyperboloid distributions has sample space the hyperboloid upper sheet and natural conic parameter space.}
\label{fig:hyperboloid}
\end{figure}

Let the Minkowski inner product~\cite{cho2019large} be 
\[ [(x_0,x_1, \cdots, x_d), (y_0,y_1, \cdots, y_d)] := x_0 y_0 - x_1 y_1 - \cdots - x_d y_d. \]
We have $\mathbb{L}^d =\{x\in\bbR^{d+1}\ :\ [x,x]=1, x_0 > 0\}$ and $\mathbb{L}^d \subset \Theta_{\mathbb{L}^d}$.

Define a map $\varphi : \mathbb{R}^d \to \mathbb{L}^d$ by $\varphi(x_1, \cdots, x_d) := \left(\sqrt{1+\sum_{i=1}^{d} x_i^2}, x_1, \cdots, x_d\right)$. 
This is a homeomorphism between $\mathbb{L}^d$ and $\mathbb{R}^d$. 
The inverse map of $\varphi$ is the projection map $(x_0,x_1, \cdots, x_d) \mapsto (x_1, \cdots, x_d)$. 
We identify $\mathbb{L}^d$ with $\mathbb{R}^d$ by this map.

Now we define the {\em hyperboloid distribution} by following Section~7 in~\cite{Barndorff-Nielsen1978hyperbolic,barndorff1982hyperbolic,barndorff1981hyperbolic}. 
Hereafter, for ease of notation, we let $|\theta| := [\theta, \theta]^{1/2}, \ \theta \in \Theta_{\mathbb{L}^d}$ and $\widetilde{x} := \varphi(x)$. 
For $\theta \in \Theta_{\mathbb{L}^d}$, 
we define a probability measure $P_{\theta}$ on $\mathbb{R}^d$ by 
\begin{equation}\label{eq:hm}
P_{\theta}(\mathrm{d}x_1 \cdots \mathrm{d}x_d) := c_d (|\theta|)\exp(-[\theta,\widetilde{x}]) \mu(\mathrm{d}x_1 \cdots \mathrm{d}x_d), 
\end{equation}
where we let 
\begin{equation}\label{eq:def-cd} 
c_d (t) := \frac{t^{(d-1)/2}}{2(2\pi)^{(d-1)/2} K_{(d-1)/2}(t)}, \ \ \ t > 0, 
\end{equation}
\[ \widetilde{x} := \left(\sqrt{1+\sum_i x_i^2}, x_1, \cdots, x_d\right), \]
and 
\[ \mu(\mathrm{d}x_1 \cdots \mathrm{d}x_d) := \frac{1}{\sqrt{1+ \sum_i x_i^2}} \mathrm{d}x_1 \cdots \mathrm{d}x_d, \]
where $K_{(d-1)/2}$ is the modified Bessel function of the second kind with index $(d-1)/2$ (~\cite{Gradshteyn2015}).  
We can regard this as a Borel probability measure on $\mathbb{L}^d$. 
The 1D hyperboloid distribution was first introduced in statistics in 1977~\cite{barndorff1977exponentially} to model the log-size distributions of particles from aeolian sand deposits, but the 3D hyperboloid distribution was 
 later found already studied in statistical physics in 1911~\cite{juttner1911maxwellsche}.
 The 2D hyperboloid distribution was investigated in 1981~\cite{blaesild1981two}.

We can rewrite the probability density function of the hyperboloid distribution of Eq.~\eqref{eq:hm} in the canonical form of an exponential family of order~\cite{EF-2014} $m=d+1$ as follows:
\begin{eqnarray}
p_{\theta}(x_1, \ldots, x_d) &=& c_d (|\theta|)\exp(-\MInner{\theta}{\widetilde{x}})\frac{1}{\sqrt{1+\sum_i x_i^2}},\\
&=&
\exp\left(\sum_{i=1}^d\theta_i\widetilde{x}_i-\theta_0\widetilde{x}_0+\log c_d(|\theta|)-\frac{1}{2}\log\left(1+\sum_{i=1}^d x_i^2\right)\right),\\
&=& \exp({t(x)}^\top{\theta}-F(\theta)+k(x)),
\end{eqnarray}
where $t(x)=\left(-\sqrt{1+\sum_i x_i^2},x_1,\ldots,x_d\right)$ is the vector of sufficient statistics (linearly independent) and $F(\theta)=-\log c_d\left(\sqrt{\theta_0^2-\sum_{i=1}^d\theta_i^2}\right)$ is the log-normalizer 
 (cumulant function) and $k(x)=-\frac{1}{2}\log\left(1+\sum_{i=1}^d x_i^2\right)$ is the auxiliary carrier term.
 The natural parameter space is  $\Theta_{\mathbb{L}^d}$.
When $d=2$, the order is $m=d+1=3$, which coincides with the order of the exponential family of Poincar\'e distributions.  
 
 It can be shown that
\begin{equation}\label{eq:Fnormalize}
F(\theta) = \log K_{(d-1)/2} (|\theta|) - \frac{d-1}{2} \log |\theta| + \log(2(2\pi)^{(d-1)/2}).
\end{equation}

Table~\ref{tab:hyperboloidef} summarizes the canonical decomposition of the exponential family of hyperboloid distributions.

\begin{table}
\centering
\scalebox{0.8}[0.8]{
\begin{tabular}{|l|l|}\hline
Hyperboloid pdf & $p_{\theta}(x) = c_d (|\theta|)\exp(-\MInner{\theta}{\widetilde{x}})\frac{1}{\sqrt{1+\sum_{i=1}^{d} x_i^2}}$\\ \hline
Exponential family & $p_{\theta}(x) = \exp\left(\MInner{t(x)}{\theta}-F(\theta)+k(x)\right)$\\
Sufficient statistic & $t(x)=\left(-\sqrt{1+\sum_i x_i^2},x_1,\ldots,x_d\right)$\\
Auxiliary carrier term & $k(x)=-\frac{1}{2}\log\left(1+\sum_{i=1}^d x_i^2\right)$\\
Cumulant function & $F(\theta) = \log K_{(d-1)/2} (|\theta|) - \frac{d-1}{2} \log |\theta| + \log(2(2\pi)^{(d-1)/2})$\\ \hline
\end{tabular}
}
\caption{Canonical decomposition of the exponential family of hyperboloid distributions.
Here, $|\theta| := [\theta, \theta]^{1/2}$ for $\theta \in \Theta_{\mathbb{L}^d}$}\label{tab:hyperboloidef}
\end{table}
 
 Since the Fisher information matrix (FIM) of an exponential family with log-normalizer $F(\theta)$ is $\nabla^2 F(\theta)$, we get the FIM for the hyperboloid distributions as 
\begin{equation*}
I(\theta) =-\nabla^2 \log c_d\left( |\theta| \right) =-\nabla^2 \log c_d\left(\sqrt{\theta_0^2-\sum_{i=1}^d\theta_i^2}\right).
\end{equation*}

Let $d=2$. 
We remark that 
\begin{equation}\label{eq:explicitK} 
K_{1/2}(z) = \sqrt{\frac{\pi}{2z}} e^{-z}, \ z > 0. 
\end{equation}
Then, for $\theta = (\theta_0, \theta_1, \theta_2) \in \Theta_{\mathbb L^2}$, 
\[ P_{\theta}(\mathrm{d}x_1 \mathrm{d}x_2) = \frac{\sqrt{\theta_0^2 - \theta_1^2 - \theta_2^2}\exp(\sqrt{\theta_0^2 - \theta_1^2 - \theta_2^2})}{2\pi} \frac{\exp(-(\theta_0\sqrt{1+x_1^2+x_2^2} - \theta_1 x_1 - \theta_2 x_2))}{\sqrt{1+x_1^2+x_2^2}} \ \mathrm{d}x_1 \mathrm{d}x_2. \]

The log-normalizer of this family of 2D hyperboloid distributions (order $m=3$) is
\begin{equation*}
F(\theta)=-\frac{1}{2}\log\frac{\theta_0^2 - \theta_1^2 - \theta_2^2}{\sqrt{2\pi}} - \sqrt{\theta_0^2 - \theta_1^2 - \theta_2^2}.
\end{equation*}
The gradient is
$$ 
\eta=\nabla F(\theta)= \left(\frac{1}{|\theta|} + \frac{1}{|\theta|^2}\right)\left[
\begin{array}{l}
-\theta_0\cr
\ \, \theta_1\cr
\ \, \theta_2\cr
\end{array}
\right].
$$

The Kullback-Leibler divergence between two densities $p_{\theta_1}$ and $p_{\theta_2}$ of an exponential family expressed using the standard inner product $\inner{\cdot}{\cdot}$ is:
$$
D_\KL[p_{\theta_1}:p_{\theta_2}]=F(\theta_2)-F(\theta_1)-\inner{\theta_2-\theta_1}{\nabla F(\theta_1)}.
$$

Thus, for $d=2$, the FIM for the local coordinate $\theta = (\theta_0, \theta_1, \theta_2)$ can be calculated in closed-form using symbolic computing: 
\begin{align}\label{eq:fim-hyp-2D} 
I(\theta) &= \left(\frac{\partial^2 F}{\partial \theta_i \partial \theta_j} (\theta) \right)_{i,j=0,1,2}\\
&= \frac{1}{|\theta|^4}
\begin{bmatrix}
(2+|\theta|)\theta_0^2 - |\theta|^2 (1+|\theta|) & -(2+|\theta|)\theta_0 \theta_1  & -(2+|\theta|)\theta_0 \theta_2 \\
-(2+|\theta|)\theta_0 \theta_1 & (2+|\theta|)\theta_1^2 + |\theta|^2 (1+|\theta|) & (2+|\theta|)\theta_1 \theta_2 \\
-(2+|\theta|)\theta_0 \theta_2 & (2+|\theta|)\theta_1 \theta_2 & (2+|\theta|)\theta_2^2 + |\theta|^2 (1+|\theta|)
\end{bmatrix}.
\end{align}
This is simpler than \eqref{eq:fim-Poincare-2D} for the Poincar\'e distribution. 
\cite[p.126]{barndorff1989decomposition} gives the FIM for the local coordinate $(|\theta|, \theta_1/|\theta|, \theta_2/|\theta|)$ different\footnote{Eq.~\eqref{eq:fim-hyp-2D} does not hold for this coordinate.} from ours, but its simple expression of \cite[p.126]{barndorff1989decomposition} is useful to show that the statistical manifold is Hadamard but not Einstein.

The main purpose of this section is to give explicit formulae for the $f$-divergences between  hyperboloid distributions in \S\ref{sec:fdiv-hyp} and then deal with mixtures of the hyperboloid distributions in \S\ref{sec:mixuda}. Finally, in \S\ref{sec:correspondence}, we deal with relationships between the Poincar\'e distribution and the hyperboloid distribution. 

\subsection{Statistical $f$-divergences between hyperboloid distributions}\label{sec:fdiv-hyp}

In this subsection, we show that every $f$-divergence between hyperboloid distributions is a function of a triplet consisting of certain functions defined by the parameters, and then give explicit formulae for the Kullback-Leibler divergence, the squared Hellinger divergence and the Neyman chi-squared divergence.  

Now we consider group actions on the space of parameters $\Theta_{\mathbb{L}^d}$. 
Let the indefinite special orthogonal group be 
\[ \SO (1,d) := \left\{A \in \SL(d+1, \mathbb{R}) : [Ax, Ay] = [x,y] \ \forall x, y \in \mathbb{R}^{d+1}  \right\}, \]
and 
\[ \SO_0 (1,d) := \left\{A \in \SO(1,d) :  A(\mathbb{L}^d) = \mathbb{L}^d  \right\}. \]
 
An action of $\SO_0 (1,d)$ on $\Theta_{\mathbb{L}^d}$ is defined by 
$A\theta$, which is just the product of matrices. 
By the map $\varphi(x) := \widetilde{x}$, $x \in \mathbb{R}^d$, 
we can regard $\mu$ as an infinite measure on $\mathbb{L}^d$ and denote it by $\widetilde{\mu}$. 
This is the push-forward measure of $\mu$ by $\varphi$. 
Then, by \cite[Section C]{HyperboloidDistribution-1981}, 
$\widetilde{\mu}$ is an $\SO(1,d)$-invariant Borel measure on $\mathbb{L}^d$, that is, $\widetilde{\mu} = \widetilde{\mu} \circ A^{-1}$.  
Let $\widetilde A : \mathbb{R}^d \to \mathbb{R}^d$ be a conjugate map of $A$ by $\varphi$, 
namely, a map defined by $\widetilde A := \varphi^{-1} \circ A \circ \varphi$. 
Then, $\mu = \mu \circ \widetilde{A}^{-1}$. 

By using this invariance, $|A\theta| = [A\theta, A\theta]^{1/2} = [\theta, \theta]^{1/2} = |\theta|$, and 
the property that $[A\theta, A\widetilde{x}] = [\theta, \widetilde{x}], A \in SO_0 (1,d)$, 
we see that 
\[ P_{\theta} \circ \widetilde{A}^{-1} = P_{A\theta}, \ \ \ A \in \SO_0 (1,d), \theta \in \Theta_{\mathbb{L}^d}, \]
where $P_{\theta} \circ \widetilde{A}^{-1}$ and $P_{A\theta}$ are probability measures on $\mathbb{R}^d$. 
Hereafter we identify $\widetilde{\mu}$ with $\mu$.  

Thus we see that the family of hyperboloid distributions is $\SO_0 (1,d)$-invariant by the pushforward induced by the action. 

The hyperboloid distributions for $d=3$ were considered in~\cite{massam1989exact}.

An action of $\SO_0 (1,d)$ on $\left(\Theta_{\mathbb{L}^d}\right)^2$ is defined by 
$$
\SO_0 (1,d) \times \left(\Theta_{\mathbb{L}^d}\right)^2 \ni \left(A,(\theta,\theta^{\prime})\right) \mapsto (A\theta, A\theta^{\prime}) \in \left(\Theta_{\mathbb{L}^d}\right)^2. 
$$

\begin{proposition}[Maximal invariant]\label{prop:max-inv-s-a-s}
The mapping $\left(\Theta_{\mathbb{L}^d}\right)^2 \ni (\theta,\theta^{\prime}) \mapsto \left([\theta, \theta],[\theta^{\prime},\theta^{\prime}],[\theta,\theta^{\prime}]\right) \in (\bbR_{>0})^2 \times \mathbb{R}$ is a maximal invariant for the action of $\SO_0 (1,d)$ on $\left(\Theta_{\mathbb{L}^d}\right)^2$ defined above. 
\end{proposition}

In the following proof, all vectors are column vectors. 

\begin{proof}
It is clear that the map is invariant with respect to the group action. 
Assume that 
$$\left([\theta^{(1)}, \theta^{(1)}],[\theta^{(2)},\theta^{(2)}],[\theta^{(1)},\theta^{(2)}]\right) 
= \left(\left[\widetilde{\theta^{(1)}}, \widetilde{\theta^{(1)}}\right],\left[\widetilde{\theta^{(2)}},\widetilde{\theta^{(2)}}\right], \left[\widetilde{\theta^{(1)}},\widetilde{\theta^{(2)}}\right]\right).$$ 
Let 
$$\psi_i := \frac{\theta^{(i)}}{\left|\theta^{(i)}\right|}, \  \widetilde{\psi_i} := \frac{\widetilde{\theta^{(i)}}}{\left|\widetilde{\theta^{(i)}}\right|},\  \ i = 1,2.$$ 
Then, 
$[\psi_1,\psi_2] = \left[\widetilde{\psi_1},\widetilde{\psi_2}\right]$. 

We first consider the case that $\psi_1 = \widetilde{\psi_1} = (1,0,\cdots,0)^{\top}$. 
Let $\psi_i = (x_{i0}, \cdots, x_{id})^{\top}, \widetilde{\psi_i} = \left(\widetilde{x_{i0}}, \cdots, \widetilde{x_{id}}\right)^{\top}, \ i = 1,2$. 
Then, $x_{20} = \widetilde{x_{20}} > 0$, 
$x_{21}^2 + \cdots + x_{2D}^2 = \widetilde{x_{21}}^2 + \cdots + \widetilde{x_{2D}}^2$ and hence there exists a special orthogonal matrix $P$ such that $P(x_{21}, \cdots, x_{2D})^{\top} = \left(\widetilde{x_{21}}, \cdots, \widetilde{x_{2D}}\right)^{\top}$. 
Let $A := \begin{bmatrix} 1 & 0 \\ 0 & P \end{bmatrix}$. 
Then, $A \in \SO_0 (1,d)$, $A\psi_1 = (1,0,\cdots,0)^{\top} = \widetilde{\psi_1}$ and $A \psi_2 = \widetilde{\psi_2}$. 

We second consider the general case. 
Since the action of $SO_0 (1,d)$ on $\mathbb{L}^d$ defined by $(A, \psi) \mapsto A\psi$ is transitive, 
there exist $A, B \in SO_0 (1,d)$ such that $A\psi_1 = B\widetilde{\psi_1} = (1,0,\cdots,0)^{\top}$. 
Thus this case is attributed to the first case. 
\end{proof}

Now we have that 
\begin{theorem}[Canonical terms of the $f$-divergences between the hyperboloid distributions]\label{thm:fdivhyperboloid}
Every $f$-divergence between $p_{\theta}$ and $p_{\theta^{\prime}}$ is invariant with respect to the action of $\SO_0(1,d)$, 
and is a function of the triplet 
$\left([\theta, \theta],[\theta^{\prime},\theta^{\prime}],[\theta,\theta^{\prime}]\right)$. 
\end{theorem}

There is a clear geometric interpretation of this. 
The side-angle-side theorem for triangles in Euclidean geometry states that if two sides and the included angle of one triangle are equal to two sides and the included angle of another triangle, then, the triangles are congruent. 
This is also true for the hyperbolic geometry and it corresponds to Proposition \ref{prop:max-inv-s-a-s} above. 
Every $f$-divergence is determined by the triangle formed by a pair of the parameters $(\theta, \theta^{\prime})$ when $f$ is fixed. 

The statement of Theorem \ref{thm:fdivhyperboloid} also holds for a deformed family of the hyperboloid distribution. 
See \cite[Theorem 4]{nielsen2023i} for details.

\begin{proposition}[Examples of divergences between hyperboloid distributions]\label{prop:exfdivhyperboloid}
We have that  \\
(i) (Kullback-Leibler divergence) 
\[ D_{\KL} [p_{\theta}:p_{\theta^{\prime}}] \]
\[=  \log \left(\frac{K_{(d-1)/2}(|\theta^{\prime}|)}{K_{(d-1)/2}(|\theta|)}\right) + \frac{d-1}{2}\left(\log \left(\frac{|\theta|}{|\theta^{\prime}|}\right) +\frac{[\theta,\theta^{\prime}]}{[\theta,\theta]} -1 \right) + \frac{[\theta,\theta-\theta^{\prime}] K_{(d-1)/2}^{\prime}(|\theta|)}{|\theta| K_{(d-1)/2}(|\theta|)}.   \]
(ii) (squared Hellinger divergence) 
Let $f(u) = (\sqrt{u} - 1)^2 / 2$. 
Then, 
\[ D_f [p_{\theta}:p_{\theta^{\prime}}] = 1 - 2^{(d-1)/2}\frac{|\theta|^{(d-1)/4}|\theta^{\prime}|^{(d-1)/4} K_{(d-1)/2}(|\theta+\theta^{\prime}|/2)}{|\theta+\theta^{\prime}|^{(d-1)/2}\sqrt{K_{(d-1)/2}(|\theta|)K_{(d-1)/2}(|\theta^{\prime}|)}}.\]
(iii) (Neyman chi-squared divergence)
Let $f(u) := (u-1)^2$. 
Then, 
\[ D_f [p_{\theta}: p_{\theta^{\prime}}] 
= \begin{cases} \frac{|\theta^{\prime}|^{d-1} K_{(d-1)/2}(|2\theta^{\prime}-\theta|) K_{(d-1)/2}(|\theta|)}{|2\theta^{\prime} -\theta|^{(d-1)/2} |\theta|^{(d-1)/2} K_{(d-1)/2}(|\theta^{\prime}|)^2} - 1 & 2\theta^{\prime} - \theta \in \Theta_{\mathbb{L}^d} \\ +\infty & \textup{ otherwise} \end{cases}. \] 
\end{proposition}

\begin{proof}
(i) 
We see that 
\[ D_{\KL} [p_{\theta}:p_{\theta^{\prime}}] = F(\theta^{\prime}) - F(\theta) - \inner{\nabla F(\theta)}{\theta^{\prime} - \theta}. \] 
We also see that for every differentiable function $g$, 
\[ \inner{\nabla g(|\theta|)}{\theta^{\prime}-\theta}  = \frac{g^{\prime}(|\theta|)}{|\theta|} [\theta,\theta^{\prime}-\theta]. \]

By this and Eq.~\eqref{eq:Fnormalize}, 
we have assertion (i). 

(ii) 
We remark that if $\theta, \theta^{\prime} \in \Theta_{\mathbb{L}^d}$, then, $(\theta+\theta^{\prime})/2 \in \Theta_{\mathbb{L}^d}$.
We see that 
\[ \int \sqrt{p_{\theta}(x)p_{\theta^{\prime}}(x)} \dx = \frac{c_d (|\theta|)^{1/2}c_d (|\theta^{\prime}|)^{1/2}}{c_d (|\theta+\theta^{\prime}|/2)}. \] 
Now recall the definition of $c_d$ in Eq.~\eqref{eq:def-cd}. 

(iii) 
Assume that $2\theta^{\prime} - \theta \in \Theta_{\mathbb{L}^d}$. 
Then, 
\[ \int \left(\frac{p_{\theta^{\prime}}(x)}{p_{\theta}(x)}-1\right)^2 p_{\theta}(x) \dx = \int \frac{p_{\theta^{\prime}}(x)^2}{p_{\theta}(x)} \dx -1 = \frac{c_d(|\theta^{\prime}|)^2}{c_d(|\theta|) c_d(|2\theta^{\prime} - \theta|)}-1. \]

Assume that $2\theta^{\prime} - \theta \notin \Theta_{\mathbb{L}^d}$. 
Let $\zeta = (\zeta_0, \cdots, \zeta_d) := 2\theta^{\prime} - \theta$. 
We first consider the case that $\zeta_0 \ne 0$. 
Let 
$\zeta(t) := (t|\zeta_0|, \zeta_1, \cdots, \zeta_d)$ for $t > 0$.  
Let 
\[ G(t) := \int_{\mathbb{R}^d} \exp(-[\zeta(t),\widetilde{x}]) \mu(dx), \ t > 0. \] 
We remark that $G(t)$ can be infinite.  
If $t > \sqrt{\sum_{i=1}^{d} \zeta_i^2}/|\zeta_0|$, then, $\zeta(t) \in \Theta_{\mathbb{L}^d}$, and hence, $G(t) = 1/c_d(|\zeta(t)|) < +\infty$. 

By \cite[Eq.~(8.432.5)]{Gradshteyn2015}, 
\[ K_{(d-1)/2}(t) = \frac{\Gamma(d/2) (2t)^{(d-1)/2}}{\sqrt{\pi}} \int_{0}^{\infty} \frac{\cos x}{(x^2 + t^2)^{d/2}} dx.  \]
By this and \eqref{eq:def-cd},
\[ c_d (t) = \frac{1}{2^d \pi^{d/2 - 1} \Gamma(d/2)} \frac{1}{\int_{0}^{\infty} \frac{\cos x}{(x^2 + t^2)^{d/2}} dx}. \]

By the monotone convergence theorem, 
\[ \int_{0}^{\pi/4} \frac{\cos x}{(x^2 + t^2)^{d/2}} dx \ge \frac{1}{\sqrt{2}} \int_{0}^{\pi/4} \frac{1}{(x^2 + t^2)^{d/2}} dx \to \frac{1}{\sqrt{2}} \int_{0}^{\pi/4} \frac{1}{x^d} dx = +\infty, \ t \to +0. \]
We also see that 
\[ \left| \int_{\pi/4}^{\infty} \frac{\cos x}{(x^2 + t^2)^{d/2}} dx \right| \le \int_{\pi/4}^{\infty} \frac{1}{x^d} dx < +\infty. \]

Therefore, $\displaystyle \lim_{t \to +0} c_d (t) = 0$. 
By this and $|\zeta(t)| = \sqrt{t^2 \zeta_0^2 - \sum_{i=1}^{d} \zeta_i^2}$, 
\[ \lim_{t \to \sqrt{\sum_{i=1}^{d} \zeta_i^2}/|\zeta_0| +0} G(t) =  \lim_{t \to \sqrt{\sum_{i=1}^{d} \zeta_i^2}/|\zeta_0| +0} \frac{1}{c_d (|\zeta(t)|)} =  +\infty.\]
Since $2\theta^{\prime} - \theta \notin \Theta_{\mathbb{L}^d}$, $\sqrt{\sum_{i=1}^{d} \zeta_i^2}/|\zeta_0| \ge 1$. 
Since $[\zeta(t),\widetilde{x}] = t|\zeta_0|\sqrt{1+\sum_{i=1}^{d} x_i^2} - \sum_{i=1}^{d} x_i \zeta_i$ and $|\zeta_0| > 0$, 
$[\zeta(t),\widetilde{x}]$ is increasing with respect to $t$, and hence, 
$G$ is a decreasing function. 
Therefore, $G(1) = +\infty$.
Since $\left[\zeta(1), \widetilde{x}\right] \ge \left[2\theta^{\prime} - \theta, \widetilde{x}\right]$, 
\[ \int_{\bbR^d} \exp\left(-[2\theta^{\prime} - \theta, \widetilde{x}]\right) \mu(dx) \ge G(1) = +\infty. \]
Thus, 
\[ \int_{\bbR^d} \frac{p_{\theta^{\prime}}(x)^2}{p_{\theta}(x)} \dx = \frac{c_d(|\theta^{\prime}|)^2}{c_d (|\theta|)} \int_{\bbR^d} \exp\left(-[2\theta^{\prime} - \theta, \widetilde{x}]\right) \mu(dx) =  +\infty. \]

Second, we consider the case that $\zeta_0 = 0$. 
If $(\zeta_1, \cdots, \zeta_d) = (0, \cdots, 0)$, then, $2\theta^{\prime} - \theta = 0$, and then, 
$\int_{\bbR^d} \exp\left(-[2\theta^{\prime} - \theta, \widetilde{x}]\right) \mu(dx) = \mu(\bbR^d) =  +\infty.$
Assume that $(\zeta_1, \cdots, \zeta_d) \ne (0, \cdots, 0)$. 
For $0 < \epsilon < \sqrt{\sum_{i=1}^{d} \zeta_i^2}$, let $\theta^{\prime}_{\epsilon} := \theta^{\prime} + (\epsilon/2, 0, \dots, 0)$. 
Then, $\theta^{\prime}_{\epsilon} \in \Theta_{\mathbb{L}^d}$ and 
$2\theta_{\epsilon}^{\prime} - \theta = (\epsilon, \zeta_1, \cdots, \zeta_d) \notin \Theta_{\mathbb{L}^d}$. 
Since $\epsilon > 0$, 
$\int_{\bbR^d} \exp\left(-[2\theta_{\epsilon}^{\prime} - \theta, \widetilde{x}]\right) \mu(dx) = +\infty$. 
$\exp\left(-[2\theta_{\epsilon}^{\prime} - \theta, \widetilde{x}]\right)$ increases as $\epsilon$ decreases, and furthermore, $\exp\left(-[2\theta_{\epsilon}^{\prime} - \theta, \widetilde{x}]\right) \to \exp\left(-[2\theta^{\prime} - \theta, \widetilde{x}]\right)$ as $\epsilon \to +0$. 
Therefore, by the monotone convergence theorem, 
\[ \int_{\bbR^d} \exp\left(-[2\theta^{\prime} - \theta, \widetilde{x}]\right) \mu(dx) = \lim_{\epsilon \to +0} \int_{\bbR^d} \exp\left(-[2\theta_{\epsilon}^{\prime} - \theta, \widetilde{x}]\right) \mu(dx) = +\infty. \]
This completes the proof. 
\end{proof}

By this assertion and Eq.~\eqref{eq:explicitK}, 
we have that 
\begin{corollary}\label{cor:exfdiv-2d}
Let $d=2$. Then, \\ 
(i) (Kullback-Leibler divergence) 
Let $f(u) = -\log u$.
Then, 
\[ D_f [p_{\theta}:p_{\theta^{\prime}}] 
= \log \left(\frac{|\theta|}{|\theta^{\prime}|}\right) - |\theta^{\prime}| + \frac{[\theta,\theta^{\prime}]}{[\theta,\theta]} + \frac{[\theta,\theta^{\prime}]}{|\theta|} -1.   \]
(ii) (squared Hellinger divergence) 
Let $f(u) = (\sqrt{u} - 1)^2 / 2$. 
Then, 
\[ D_f [p_{\theta}: p_{\theta^{\prime}}] 
= 1 - \frac{2|\theta|^{1/2}|\theta^{\prime}|^{1/2} 
\exp\left(|\theta|/2 + |\theta^{\prime}|/2\right)}{ |\theta+\theta^{\prime}|  \exp\left(|\theta+\theta^{\prime}|/2\right)}.\]
(iii) (Neyman chi-squared divergence)
Let $f(u) := (u-1)^2$. 
Then, 
\[ 
D_f [p_{\theta}: p_{\theta^{\prime}}] 
= \begin{cases} \frac{|\theta^{\prime}|^2 \exp(2|\theta^{\prime}|)}{|\theta||2\theta^{\prime}-\theta|\exp(|\theta| + |2\theta^{\prime}-\theta|)} - 1 & 2\theta^{\prime} - \theta \in \Theta_{\mathbb{L}^2} \\ +\infty & \textup{ otherwise}\end{cases}.  \] 
\end{corollary}

\begin{remark}\label{rmk:NeyChi}
(Foliation) For $t > 0$, let 
\[ \Theta(t) := \{\theta \in \Theta_{\mathbb L^d} : |\theta| = t\}. \]
Then, 
by noting that $[\theta, \theta^{\prime}] = [\theta^{\prime},\theta]$, 
every $f$-divergence is symmetric on $\Theta(t)$. 
For example, if $t=1$ and $d=2$, then, for $f(u) = -\log u$, 
\[  D_f [p_{\theta}: p_{\theta^{\prime}}] = 2([\theta,\theta^{\prime}]-1)\]
on $\Theta(1)$. 
\end{remark}

\subsection{Hyperboloid mixtures as universal density approximator}\label{sec:mixuda}

In this subsection, we show that the set of mixtures of the hyperboloid distributions is dense in the space of the probability distribution on $\mathbb{L}^d$ with the weak-star topology. 

The set of the $d$-dimensional {\it hyperboloid mixture model} is given by 
\[ \mathcal{M}_{\mathbb L^d} := \left\{\sum_{i=1}^{n} w_i P_{\theta_i} : \theta_1, \dots, \theta_n \in \Theta_{\mathbb L^d}, w_i \ge 0, \sum_{i=1}^{n} w_i = 1 \right\}. \]

The following shows that the hyperboloid mixture model has a universal property as the location-scale model.  
\begin{theorem}\label{thm:densemixtures}
The set of statistical mixtures $\mathcal{M}_{\mathbb L^d}$ is dense in the space of the probability distribution on $\mathbb{L}^d$ for the topology of the weak convergence.  
\end{theorem}

\begin{proof}
Step 1. We first remark that the set of the Dirac masses $\{\delta_x\}_{x \in \mathbb{R}^d}$ is dense in the space of the probability distribution on $\mathbb{R}^d$ for the topology of the weak convergence, by applying the Krein-Milman theorem to the space of the probability measures on a compact subset of $\mathbb{R}^d$ and using the fact that $\mathbb{R}^d$ is $\sigma$-compact.  

We now give some details. 
Let $X$ be a metric space. 
We denote the ball with center $x \in X$ and radius $r > 0$ by $B(x,r)$. 
Let $C(X)$ be the Banach space of bounded continuous functions on $X$ with the supremum norm. 
Let $C(X)^*$ be the dual Banach space of $C(X)$ with the operator norm. 
Let $B_{C(X)^*}$ be the unit closed ball of $C(X)^*$. 
Let $\mathcal{P}(X)$ be the set of Borel probability measures on $X$. 
We define a map $\Phi : \mathcal{P}(X) \to B_{C(X)^*}$ by 
\[\Phi(\mu)(f) := \int_X f d\mu, \ \ f \in C(X). \]
By this map we identify $\mathcal{P}(X)$ with a subset of $B_{C(X)^*}$. 
Assume that $X$ is compact. 
Then, $\mathcal{P}(X)$ is compact with respect to the weak-star topology. 

Let $\textup{ex}(\mathcal{P}(X))$ be the extreme points of $\mathcal{P}(X)$ in $C(X)^*$. 
For the sake of completeness, we will show the well-known fact that $\textup{ex}(\mathcal{P}(X)) = \{\delta_x : x \in X\}$, where $\delta_x$ denotes the Delta mass at $x$. 
Let $x_0 \in X$. 
Assume that there exist $P_1, P_2 \in \mathcal{P}(X)$ and $a \in (0,1)$ such that  $\delta_{x_0} = (1-a)P_1 + aP_2$. 
For each $n \ge 1$, there exists $f_n \in C(X)$ such that $f_n(x_0) = 1$, $f_n(x) = 0$ for $x \notin B(x_0, 1/n)$, and $0 \le f_n \le 1$. 
Then, by the Lebesgue convergence theorem, 
$\int_X f_n dP_i  \to P_i (\{x_0\})$, $n \to \infty$, $i=1,2$. 
Hence, $1  = (1-a) P_1 (\{x_0\}) + a P_2 (\{x_0\})$, and hence, $P_1 = P_2 = \delta_{x_0}$. 
Thus we see that $\delta_{x_0} \in \textup{ex}(\mathcal{P}(X))$. 

Let $\mu \in \mathcal{P}(X)$ and assume that $\mu \ne \delta_x$ for every $x \in X$. 
Then, $\textup{supp}(\mu)$ contains at least two points and we denote them by $x_1 \ne x_2$. 
Since $X$ is compact Hausdorff, 
there exists two disjoint closed balls $B_1$ and $B_2$ such that $x_i \in B_i$, $i=1,2$. 
Then, $\mu(B_1) > 0$ and $\mu(B_1^c) \ge \mu(B_2) > 0$. 
We see that for every Borel subset $A$ of $X$, 
\[ \mu(A) = \mu(B_1) \frac{\mu(A  \cap B_1)}{\mu(B_1)} + (1-\mu(B_1)) \frac{\mu(A \cap B_1^c)}{\mu(B_1^c)}. \]
Let $a := \mu(B_1^c)$ and $P_1 (A) := \frac{\mu(A  \cap B_1)}{\mu(B_1)} \in \mathcal{P}(X)$ and $P_2 (A) := \frac{\mu(A \cap B_1^c)}{\mu(B_1^c)}  \in \mathcal{P}(X)$. 
Then, $\mu = (1-a)P_1 + aP_2$, and $P_1 \ne \mu$. 
Hence, $\mu \notin  \textup{ex}(\mathcal{P}(X))$. 

Now, by the Krein-Milman theorem \cite[Section 14.6]{Royden2010}, 
the convex hull of $\textup{ex}(\mathcal{P}(X)) = \{\delta_x : x \in X\}$ 
is dense in $\mathcal{P}(X)$ with respect  to the weak-star topology. 

Finally, we apply this result to Euclidean space. 
Let $\Delta^{k} := \left\{x = (x_i)_i \in \bbR^{k+1} \ \middle| \ x_i \ge 0, \sum_{i=1}^{k+1} x_i = 1 \right\}$ for $k \ge 1$ for $k \ge 1$. 
Let 
$$\mathcal{F}_N := \left\{\sum_{i=1}^{m} a_i \delta_{x_i} : m \ge 1, (a_1, \cdots, a_m) \in \Delta^{m-1}, x_1, \cdots, x_m \in \overline{B(0,N)} \right\}, N \ge 1.$$
Then, by the Krein-Milman theorem, the convex hull of $\mathcal{F}_N = \textup{ex}\left( \mathcal{P}(\overline{B(0,N)}) \right)$ is dense in $\mathcal{P}(\overline{B(0,N)})$. 
For $\mu \in \mathcal{P}(\mathbb{R}^d)$ and for sufficiently large $N$, 
let $\mu_N (A) := \dfrac{\mu\left(A \cap \overline{B(0,N)}\right)}{\mu\left(\overline{B(0,N)}\right)}$. 
Then, $\mu_N \in \mathcal{P}(\overline{B(0,N)})$ and $\mu_N \to \mu, N \to \infty$, weakly. 
Let $\mathcal{F} := \cup_{N \ge 1} \mathcal{F}_N$. 
Then, 
\[ \mathcal{F} = \left\{\sum_{i=1}^{m} a_i \delta_{x_i} : m \ge 1, (a_1, \cdots, a_m) \in \Delta^{m-1}, x_1, \cdots, x_m \in \mathbb{R}^d \right\}, \]
and $\mathcal{F}$ is dense in $\mathcal{P}(\mathbb{R}^d)$. 
Hence, 
\[ \mathcal{G} := \left\{\sum_{i=1}^{m} a_i \delta_{x_i} : m \ge 1, (a_1, \cdots, a_m) \in \Delta^{m-1},  x_1, \cdots, x_m \in \mathbb{L}^d \right\} \] 
is dense in $\mathcal{P}(\mathbb{L}^d)$.

Step 2. By Step 1, 
it suffices to show that for every $\theta \in \mathbb{L}^d$, 
$P_{t \theta}$ converges weakly to $\delta_{\theta}$ as $t \to +\infty$. 
That is, it suffices to show that for every bounded continuous function $f$ on $\mathbb{L}^d$, 
\begin{equation}\label{eq:1st-reduction} 
\int_{\mathbb L^d} f(x) \mathrm{d} P_{t \theta} (x) \to f(\theta) = \int_{\mathbb{L}^d} f(x) \delta_{\theta}(dx), \ \ t \to +\infty. 
\end{equation}

It is not difficult to show that if $\{x_n\}_n$ is a sequence of $\mathbb{L}^d$ and $\lim_{n \to \infty} [\theta, x_n] = 1$, then $\lim_{n \to \infty} x_n = \theta$.  
Hence, for every $\epsilon > 0$, there exists $\delta > 0$ such that for every $x \in \mathbb{L}^d$ with $[\theta, x] \le 1+\delta$, $|f(x) - f(\theta)| < \epsilon$. 
Let $B(\theta,\delta) := \{x \in \mathbb{L}^d : [\theta, x] \le 1 + \delta\}$. 
Then, for every $t > 0$, 
\[ \int_{B(\theta,\delta)} |f(x) - f(\theta)| \mathrm{d} P_{t \theta} (x) \le \epsilon. \]
Hence, by noting that $f$ is bounded,  it suffices to show that 
\[ P_{t \theta} (B(\theta,\delta)^c) \to 0, \ t \to +\infty. \]
If this holds, then, 
\[ \limsup_{t \to +\infty} \int_{\mathbb L^d} |f(x) - f(\theta)| \mathrm{d} P_{t \theta} (x) \le \epsilon. \]
We have Eq.~\eqref{eq:1st-reduction} if we let $\epsilon \to +0$.

We recall that 
$$ \mathrm{d}P_{t \theta}(x) = \dfrac{\exp(-t [\theta,x])}{\int_{\mathbb{L}^d} \exp(-t [\theta,x]) \mu(\dx)} \mu(\dx), \ x \in \mathbb{L}^d. $$
Now it suffices to show that 
\begin{equation}\label{eq:ratio-wts-final} 
\frac{P_{t \theta} (B(\theta,\delta)^c)}{P_{t \theta} (B(\theta,\delta))} = \frac{\int_{B(\theta,\delta)^c} \exp(-t [\theta,x]) \mu(dx)}{\int_{B(\theta,\delta)} \exp(-t [\theta,x]) \mu(dx)} \to 0, \ t \to +\infty. 
\end{equation}
We see that 
\[ \int_{B(\theta,\delta)} \exp(-t [\theta,x]) \mu(dx) \ge \exp(-t(1+\delta))\mu(B(\theta,\delta)).\]
By the Lebesgue convergence theorem, 
\[ \int_{B(\theta,\delta)^c} \exp\left(-t ([\theta,x] - 1 - \delta)\right) \mu(\dx) \to 0, \ t \to +\infty.\]
Thus we have Eq.~\eqref{eq:ratio-wts-final}.
\end{proof}

Thus statistical mixture models of hyperboloid distributions are universal density estimators of continuous pdfs in hyperbolic spaces.
This result is the equivalent of  Gaussian mixture models being universal density estimators in Euclidean spaces (\cite{Titterington1985} and \cite[Chapter 3, p65]{Goodfellow2016}).

\begin{remark}
Since the Poincar\'e distributions and hyperboloid distributions are exponential families of hyperbolic geometry, we can consider learning statistical mixtures using the corresponding Bregman soft clustering~\cite[Section 5]{banerjee2005clustering}, which implements the standard Expectation-Maximization algorithm~\cite{dempster1977maximum} for mixture density estimation. 
\end{remark}

\subsection{Correspondence principle}\label{sec:correspondence}

In this subsection, we give a correspondence between the Poincar\'e distribution and the hyperboloid distribution.   

It is well-known that there is a correspondence between the 2D Lobachevskii space $\mathbb L = \mathbb{L}^2$ and the Poincar\'e upper-half plane $\mathbb H$. 

\begin{proposition}\label{prop:corresp}
For $\theta = (a,b,c) \in \Theta_{\mathbb H} := \left\{(a,b,c) : a > 0, c > 0, ac > b^2\right\}$, 
let 
\[ \theta_{\mathbb L} := (a+c, a-c, 2b) \in \Theta_{\mathbb L}.\] 
We denote the $f$-divergence on $\mathbb{L}$ and $\mathbb H$  by $D_f^{\mathbb L}[ \cdot  : \cdot ]$ and $D_f^{\mathbb H}[ \cdot : \cdot ]$ respectively. 
Then,\\
(i) For $\theta, \theta^{\prime} \in \Theta_{\mathbb H}$, 
\begin{equation}\label{eq:corres-norm} 
|\theta_{\mathbb L}|^2 = \left[\theta_{\mathbb L},\theta_{\mathbb L}\right] = 4 |\theta|, \ |\theta^{\prime}_{\mathbb L}|^2 = \left[\theta^{\prime}_{\mathbb L},\theta^{\prime}_{\mathbb L}\right] = 4 |\theta^{\prime}|, \ \left[\theta_{\mathbb L},\theta^{\prime}_{\mathbb L}\right] = 2 |\theta| \textup{tr}(\theta^{\prime} \theta^{-1}).
\end{equation}
(ii) For every $f$ and $\theta, \theta^{\prime} \in \mathbb{H}$, 
\begin{equation}\label{eq:corres}
D_f^{\mathbb L}\left[p_{\theta_{\mathbb L}}:p_{\theta^{\prime}_{\mathbb L}}\right] = D_f^{\mathbb H}\left[p_{\theta}:p_{\theta^{\prime}}\right].
\end{equation}
\end{proposition}

For (i), at its first glance, there is an inconsistency in notation. 
However, $|\theta|$ is the Minkowski norm for $\theta \in \theta_{\mathbb L}$, and, $|\theta|$ is the determinant for $\theta \in \Theta_v$, so the notation is consistent {\it in each setting}.  
By this assertion, it suffices to compute the $f$-divergences between the hyperboloid distributions on $\mathbb L$. 

By Theorem \ref{thm:densemixtures} and the change-of-variable in Remark \ref{rmk:mde}  below, 
we have that 
\begin{corollary}\label{cor:densemixtures-Poincare}
The set of statistical mixtures  of the Poincar\'e distributiuons 
is dense in the space of the probability distribution on $\mathbb{H}$ for the topology of the weak convergence. 
\end{corollary}

By Proposition \ref{prop:corresp}, 
we have that 
\begin{table}[htbp]
    \centering
    \scalebox{0.8}[0.8]{
    \begin{tabular}{|c|c|}
       Poincar\'e  & 2D hyperboloid \\ \hline
       Theorem \ref{thm:fdivPoincare} & Theorem \ref{thm:fdivhyperboloid}\\
       Proposition \ref{prop:kld-Poincare} & Corollary \ref{cor:exfdiv-2d} (i)\\
       Proposition \ref{prop:SHNC-Poincare}  & Corollary \ref{cor:exfdiv-2d} (ii) (iii) \\
        Theorem \ref{thm:densemixtures} & Corollary \ref{cor:densemixtures-Poincare}
    \end{tabular}
    }
    \caption{Correspondences between the results for the Poincar\'e and hyperboloid models}
    \label{table:correspond}
\end{table}

\begin{remark}[modified differential entropy]\label{rmk:mde}
This correspondence does {\it not} hold for the differential entropy.  
Specifically, Eq.~\eqref{eq:corres} does not imply an assertion for the 2D hyperboloid model corresponding Proposition \ref{prop:h-Poincare} for the Poincar\'e distribution. 
Indeed, if we let 
$$
h^{\mathbb L}[p_{\theta}] 
:= \int_{\mathbb{R}^2} - \log(p_{\theta}(x,y))  p_{\theta}(x,y) \dx\dy, \theta \in \Theta_{\mathbb L},
$$
then, it can happen that $h^{\mathbb L}[p_{\theta^{\mathbb L}}] \ne h[p_{\theta}]$. 
The reason is that $\dx\dy$ is not an invariant measure. 

Now we modify the definition of the differential entropy. 
This modification was already pointed out by Jaynes \cite{Jaynes-1968}.

For $\theta \in \Theta_{\bbL}$, let $\widetilde p_{\theta}(x,y) := \sqrt{1+x^2+y^2} p_{\theta}(x,y), (x,y) \in \mathbb{R}^2$, and $\mu_{\mathbb L}(\dx\dy) := \dfrac{\dx\dy}{\sqrt{1+x^2+y^2}}$. 
Let 
$$
\widetilde{h}^{\mathbb L}[p_{\theta}] 
:= \int_{\mathbb{R}^2} - \log(\widetilde{p_{\theta}}(x,y))  \widetilde{p_{\theta}}(x,y) \mu_{\mathbb L}(\dx\dy), \theta \in \Theta_{\mathbb L}. 
$$

In the same manner, for $\theta \in \Theta_{\bbH}$, let 
$\widetilde p_{\theta}(x,y) := y^2 p_{\theta}(x,y), (x,y) \in \mathbb{R} \times \mathbb{R}_{++}$, and $\mu_{\mathbb H}(\dx\dy) := \dfrac{\dx\dy}{y^2}$. 
Let 
$$
\widetilde{h}^{\mathbb H}[p_{\theta}] 
:= \int_{\mathbb{R} \times \mathbb{R}_{++}} - \log(\widetilde{p_{\theta}}(x,y))  \widetilde{p_{\theta}}(x,y) \mu_{\mathbb H}(\dx\dy), \theta \in \Theta_{\mathbb H}. 
$$
By the change of variable
\[ \mathbb{H} \ni (x,y) \mapsto (X,Y) = \left(\frac{1-x^2-y^2}{2y}, -\frac{x}{y}\right) \in \mathbb{R}^2,  \]
and by recalling the correspondence between the parameters in Eq.~\eqref{eq:corres-norm},  
it holds that 
\[ \widetilde{p_{\theta}}(x,y) = y^2 p_{\theta}(x,y) = \sqrt{1+X^2 + Y^2} p_{\theta^{\mathbb L}} (X,Y) = \widetilde{p_{\theta^{\mathbb L}}}(X,Y),\]
and
\[ \mu_{\mathbb H}(\dx\dy) = \mu_{\mathbb L}(\dX\dY).\]

Hence, 
$$\widetilde{h}^{\mathbb L}[p_{\theta^{\mathbb L}}] 
= \widetilde{h}^{\mathbb H}[p_{\theta}] =  -F^*(\eta) = 1+\log\left(\dfrac{\pi}{\sqrt{|\theta|}}\right) = 1+\log\left(\dfrac{2\pi}{|\theta^{\mathbb L}|}\right).$$ 

It seems more difficult to obtain an explicit formula for the differential entropy for the 2D hyperboloid model. 
\end{remark}

\begin{remark}
By this corresponding principle, Eq.\eqref{eq:1st-reduction} holds also for Poincar\'e distributions. \\
Let $\theta = \mattwotwo{1}{x}{x}{x^2+y^2} \in \Theta_{\bbH}$ for $x+ iy \in \bbH$. 
Then, $P_{t\theta}$ weakly converges to $\delta_{x+iy}$ as $t \to +\infty$. 
Hence, an analog of Theorem \ref{thm:densemixtures} would hold also for Poincar\'e distributions. 
\end{remark}

\section{Numerical computations}\label{sec:numerics}
In this section, we consider numerical computations for $f$-divergences between hyperboloid distributions for $d=2$, in particular for $f$-divergences for which we do not have explicit formulae.   
The techniques we will use below are standard, and we do not introduce any new theoretical methodology for approximations.  
So they might be applicable to all dimensions, but for simplicity we deal with the case of the 2-dimensional case only.  
By the correspondence principle in Proposition \ref{prop:corresp}, they are also applicable to the Poincar\'e distributions. 
We mainly focus on the {\it total variation distance}: 
$$
D_\TV[p:q] := \int_{\bbR^2} \frac{1}{2}|p(x,y)-q(x,y)|\dx\dy.
$$ 
This is also an $f$-divergence with $f(x) = \frac{1}{2}|x-1|$. 

In below, we find triplets of a probability space, a two-dimensional distribution $(X,Y)$ on it and a function $g_{\theta,\theta^{\prime}}$ such that 
\begin{equation}\label{eq:fdiv-transformation}
D_{f}[p_{\theta}:p_{\theta^{\prime}}] = E\left[g_{\theta,\theta^{\prime}}(X,Y)\right]. 
\end{equation}

\subsection{Sampling random variates}\label{sec:samplingvar}

In this subsection, we give a procedure for generating random number  following the hyperboloid distribution. 

As in \cite{Barndorff-Nielsen1978hyperbolic}, the hyperboloid model is expressed by the normal distribution having  random mean and variance governed by the generalized inverse Gaussian distribution. 
There are two different types of randomness. 
We can generate random samples following the 2d hyperboloid distribution by first generating a positive random variable following the generalized inverse Gaussian distribution, and then substituting it into the mean vector and the covariance matrix of the 2d normal distribution, and finally generating a 2d random vector following the 2d normal distribution. 

Specifically, the hyperboloid distribution $P_{\theta}(dx_1 dx_2), \theta = (\theta_0, \theta_1, \theta_2) \in \Theta_v$, is identical with  the two-dimensional normal distribution with its mean vector $\sigma^2 (\theta_1, \theta_2) $ and covariance matrix $\sigma^2 I_2$, 
where $\sigma^2$ is also a random variable governed by a generalized inverse Gaussian distribution and its density function\footnote{The form depends on the dimension. 
For $d \ge 3$, it involves the modified Bessel function of the second kind. 
See \cite{Barndorff-Nielsen1978hyperbolic}.} 
is given by 
\begin{equation*}
\frac{\sqrt{\theta_0^2 - \theta_1^2 - \theta_2^2}\exp(\sqrt{\theta_0^2 - \theta_1^2 - \theta_2^2})}{2\pi} \frac{1}{\sqrt{x}} \exp\left(-\frac{1}{2}\left(\frac{1}{x} + (\theta_0^2 - \theta_1^2 - \theta_2^2) x \right) \right), \ x > 0.
\end{equation*}

Hence, we can obtain the random samples once we generate a random variable following the generalized inverse Gaussian distribution. 
We use the software {\tt Python} and the library {\tt Scipy} and the function {\tt scipy.stats.geninvgauss}. 
It depends on an acceptance-rejection method considered by~\cite{HormannLeydold2014}.

Now we recall $\mathbb{R}^2$ is identified with $\mathbb{L}^2$. 
If the distribution of $(X,Y)$ is the hyperboloid distribution, and 
$g_{\theta, \theta^{\prime}}(x,y) = \displaystyle \frac{1}{2}\left|\frac{p_{\theta^{\prime}}(x,y)}{p_{\theta}(x,y)}- 1\right|$, 
then, Eq.~\eqref{eq:fdiv-transformation} holds 
for $f(x) = \frac{1}{2}|x-1|$.

\subsection{Monte-Carlo method}
Here we use Monte-Carlo methods for approximations of $f$-divergences between the hyperboloid distributions.

(MC1) The first way is the {\it importance sampling}~\cite{gentle2003random}. 
It is a variance reduction technique in Monte-Carlo method by using a distribution with positive density. 
Let $p = p(x) > 0$ be a pdf of a continuous distribution supported on $\mathbb{R}$. 
Consider the scale family $\{p_{\sigma}(x) := (1/\sigma)*p(x/\sigma)\}_{\sigma > 0}$. 
Let $X$ and $Y$ be two independent distributions whose density functions are both $p(x)$. 
Let 
\[ H_{\theta,\theta^{\prime}}(x,y) := f\left(\frac{p_{\theta^{\prime}}(x,y)}{p_{\theta}(x,y)}\right)p_{\theta}(x,y) \]
and 
\begin{equation}\label{eq:newdensity-importance-sampling} 
g_{\theta,\theta^{\prime}}(x,y) := \frac{H_{\theta,\theta^{\prime}}(x,y)}{p_{\sigma}(x)p_{\sigma}(y)}, \ x,y \in \mathbb{R}. 
\end{equation}
Then, Eq.~\eqref{eq:fdiv-transformation} holds, and furthermore, we expect that by the law of large numbers, 
for large $n$, 
\[ D_f (p_{\theta}:p_{\theta^{\prime}}) = E[g_{\theta,\theta^{\prime}}(X,Y)] \approx \frac{1}{n} \sum_{i=1}^{n} g_{\theta,\theta^{\prime}}(X_i, Y_i), \]
where $X_i, Y_i, i \ge 1,$ are independent random variables whose distributions have a density function $p_{\sigma}(x)$ and we assume that $g(X_1, Y_1)$ is integrable. 

There are many candidates for $p(x)$ and $\sigma > 0$.  
However, the integrability of $g_{\theta,\theta^{\prime}}(X,Y)$ is important.

There are many choices of $p = p(x)$. 
In below, we consider two specific choices of $p = p(x)$: 
(i) the logistic distribution, and
(ii) the t-distribution with freedom 7. 
These perform well in the case of the Kullback-Leibler divergence. 

We optimize $\sigma > 0$ by using the same idea as the cross-entropy method. 
We aim at minimizing 
the variance $\textup{Var}(g_{\theta, \theta^{\prime}}(X,Y)) = E[g_{\theta, \theta^{\prime}}(X,Y)^2] - D_{f}(p_{\theta}:p_{\theta^{\prime}})^2$. 
Then, for large $n$, 
\[ E\left[g_{\theta, \theta^{\prime}}(X,Y)^2\right] = \int \frac{H_{\theta,\theta^{\prime}}(x,y)^2}{p_{\sigma}(x)p_{\sigma}(y)p(x)p(y)} p(x)p(y) \dx\dy  \approx \frac{1}{n} \sum_{i=1}^{n} \frac{H_{\theta,\theta^{\prime}}(X_i,Y_i)^2}{p_{\sigma}(X_i)p_{\sigma}(Y_i)p(X_i)p(Y_i)}. \]
We adopt $\sigma > 0$ which minimizes the above integral. 
In numerical computations, we substitute the sample mean above for the integral.  
The variances computed with the optimal $\sigma$ are much smaller than the variances computed with $\sigma = 1$. 

(MC2) The second way is a crude Monte-Carlo method by using change-of-variables. 
The correspondence 
$$(X,Y) \mapsto (x,y) = \left(\frac{X}{\sqrt{1-X^2-Y^2}},\frac{Y}{\sqrt{1-X^2-Y^2}}\right)$$ 
gives a diffeomorphism between $\mathbb D = \{(X,Y) : X^2 + Y^2 < 1\}$ and $\mathbb{R}^2$. 
By the change-of-variable formula, 
\[ D_f (p_{\theta}:p_{\theta^{\prime}}) \]
\[= \int_{\mathbb D} f\left(\frac{p_{\theta^{\prime}}(\frac{X}{\sqrt{1-X^2-Y^2}},\frac{Y}{\sqrt{1-X^2-Y^2}})}{p_{\theta}(\frac{X}{\sqrt{1-X^2-Y^2}},\frac{Y}{\sqrt{1-X^2-Y^2}})}\right)p_{\theta}\left(\frac{X}{\sqrt{1-X^2-Y^2}},\frac{Y}{\sqrt{1-X^2-Y^2}}\right) \frac{\dX \dY}{(1-X^2-Y^2)^2} \]
\[= \int_{r=0}^{1}\int_{\zeta=0}^{2\pi} g_{\theta,\theta^{\prime}}(r,\zeta) dr \frac{\mathrm{d}\zeta}{2\pi}, \]
where we let 
\[ g_{\theta,\theta^{\prime}}(r,\zeta) := 2\pi f\left(\frac{p_{\theta^{\prime}}(\frac{r\cos\zeta}{\sqrt{1-r^2}},\frac{r\sin\zeta}{\sqrt{1-r^2}})}{p_{\theta}(\frac{r\cos\zeta}{\sqrt{1-r^2}},\frac{r\sin\zeta}{\sqrt{1-r^2}})}\right)p_{\theta}\left(\frac{r\cos\zeta}{\sqrt{1-r^2}},\frac{r\sin\zeta}{\sqrt{1-r^2}}\right) \frac{r}{(1-r^2)^2}. \]
Let $U, V$ be the uniform distributions on $[0,1]$ and $[0,2\pi]$ respectively. 
Then, Eq.~\eqref{eq:fdiv-transformation} holds, and furthermore, we expect that by the law of large numbers, 
for large $n$, 
\[ D_f [p_{\theta}:p_{\theta^{\prime}}] = E[g_{\theta,\theta^{\prime}}(U,V)] \approx \frac{1}{n} \sum_{i=1}^{n} g_{\theta,\theta^{\prime}}(U_i, V_i),\]
where $U_i, V_i, i \ge 1,$ are independent random variables and $U_i$ and $V_i$ follow the uniform distribution on $[0,1]$ and $[0,2\pi]$ respectively. 
As in the first case, the integrability of $g_{\theta,\theta^{\prime}}(U_i, V_i)$ is important. 

\subsection{Monte Carlo estimators and the total variation distance}

We do not have an explicit formula for the total variation distance between the hyperboloid distributions, and here we consider numerical computations for the total variation distances. 

In this subsection, we let $f(x) = \frac{1}{2}|x-1|$ in the definition of $g_{\theta, \theta^{\prime}}$.   
Recall that Eq.~\eqref{eq:fdiv-transformation} holds. 

\subsubsection{Theoretical backgrounds}

A naive choice for $g_{\theta, \theta^{\prime}}$ can lead to the divergence of the variance. 

\begin{lemma}\label{lem:var-infinite-rng}
If $(X,Y)$ follows the hyperboloid distribution, and 
$\displaystyle g_{\theta, \theta^{\prime}}(x,y) = \frac{1}{2}\left|\frac{p_{\theta^{\prime}}(x,y)}{p_{\theta}(x,y)}- 1\right|$,
then, for some $(\theta, \theta^{\prime})$, 
\[ \textup{Var}\left(g_{\theta, \theta^{\prime}}(X,Y)\right) = +\infty. \]
\end{lemma}

\begin{proof}
Let $\theta = (4,0,0)$ and $\theta^{\prime} = (1,0,0)$. 
Then, for some large $M > 0$, 
$g_{\theta, \theta^{\prime}}(x,y)^2 p_{\theta}(x,y) \ge \exp((1+x^2+y^2)^{1/2})$ if $(x^2+y^2)^{1/2} \ge M$. 
Hence, 
\[ E\left[g_{\theta, \theta^{\prime}}(X,Y)^2\right] \ge \int_{(x^2+y^2)^{1/2} \ge M} g_{\theta, \theta^{\prime}}(x,y)^2 p_{\theta}(x,y) \dx\dy = +\infty.  \]
\end{proof}

Theoretically, the Student $t$-distribution would be a good choice for the density $p = p(x)$ in the importance sampling, because $g_{\theta, \theta^{\prime}, \sigma}(x,y)$ is bounded and hence good tail estimates hold. 
The following deals with (MC1). 

\begin{lemma}\label{lem:bdd-var-chernoff-t}
Let $\theta, \theta^{\prime} \in \Theta_{\mathbb{L}^2}$.   
Let $p(x)$ be a probability density function of a t-distribution with degree of freedom $m \ge 1$. 
Let $X_i, Y_i, i \ge 1,$ be independent random variables following the $t$-distribution. 
Let $\sigma > 0$. 
Let $g_{\theta, \theta^{\prime},\sigma}(x,y)$ be a function as in Eq.~\eqref{eq:newdensity-importance-sampling}. 
Then, \\
(i) $g_{\theta, \theta^{\prime},\sigma}(x,y)$ is bounded on $\mathbb{R}^2$, that is, $\|g_{\theta, \theta^{\prime},\sigma}\|_{\infty} < +\infty$. \\
(ii) For every $n \ge 1$, 
\begin{equation*} 
\textup{Var}\left(\frac{1}{n} \sum_{i=1}^{n} g_{\theta, \theta^{\prime},\sigma}(X_i,Y_i)\right) 
\le \frac{\|g_{\theta, \theta^{\prime},\sigma}\|_{\infty}^2}{4n}. 
\end{equation*}
(iii) For every $n \ge 1$ and every $t > 0$, 
\[ P\left(\left|\frac{1}{n} \sum_{i=1}^{n} g_{\theta, \theta^{\prime},\sigma}(X_i,Y_i) - D_{\TV} [p_{\theta}:p_{\theta^{\prime}}] \right| > t\right) \le 2\min\left\{\frac{\|g_{\theta, \theta^{\prime},\sigma}\|_{\infty}^2}{\|g_{\theta, \theta^{\prime},\sigma}\|_{\infty}^2 + 4nt^2}, \exp\left(- \frac{nt^2}{\|g_{\theta, \theta^{\prime},\sigma}\|_{\infty}^2}\right)\right\}. \] 
\end{lemma}

Contrary to the Lemma \ref{lem:var-infinite-rng} above, the variance is finite. 

\begin{proof}
(i) By applying \eqref{eq:newdensity-importance-sampling} to $f(x) = \frac{1}{2} |x-1|$,
\[ g_{\theta, \theta^{\prime}, \sigma}(x,y) = \frac{1}{p_{\sigma}(x)} \frac{1}{p_{\sigma}(y)} \frac{|\exp(-[\theta, (\sqrt{1+x^2+y^2},x,y)]) - \exp(-[\theta^{\prime}, (\sqrt{1+x^2+y^2},x,y)])|}{2\sqrt{1+\sum_{i=1}^{d} x_i^2} }  \]
\[ \le \frac{1}{p_{\sigma}(x)} \frac{1}{p_{\sigma}(y)} \left(\exp(-[\theta, (\sqrt{1+x^2+y^2},x,y)]) + \exp(-[\theta^{\prime}, (\sqrt{1+x^2+y^2},x,y)])\right). \]

Since $p(x)$ is the density function of $t$-distribution, 
$\frac{1}{p_{\sigma}(x)}$ and $\frac{1}{p_{\sigma}(y)}$ have at most polynomial growths\footnote{If $m$ is odd, then, they are polynomials. If $m$ is even, they are the square roots of polynomials.}. 
Hence, (i) follows from the comparison of the polynomial growth and the exponential growth of the functions on $\mathbb{R}^2$.  
Let $(x_1, x_2) := (x,y)$ in the following.  
Let $\theta = (\theta_0,\theta_1,\theta_2) \in \Theta_{\mathbb{L}^2}$.  
Then, for sufficiently small $\epsilon > 0$, 
$\theta_{(\epsilon)} := (\theta_0 - \epsilon,\theta_1,\theta_2) \in \Theta_{\mathbb{L}^2}$ and $\inf_{x \in \mathbb{R}^2} [\theta_{(\epsilon)}, \widetilde{x}] = |\theta_{(\epsilon)}| > 0$. 
We see that for every $n_1, n_2 \ge 0$, 
\[ |x_1|^{n_1} |x_2|^{n_2} \exp(-[\theta,\widetilde{x}]) 
= |x_1|^{n_1} |x_2|^{n_2} \exp(-[\theta_{(\epsilon)},\widetilde{x}])\exp\left(-\epsilon \sqrt{1+x_1^2+x_2^2}\right)\]
\[ \le |x_1|^{n_1} \exp(-\epsilon |x_1|/4) |x_2|^{n_2} \exp(-\epsilon |x_2|/4) \exp(-|\theta_{(\epsilon)}|).  \]

(ii) We remark that $g_{\theta, \theta^{\prime},\sigma}$ is non-negative in the case of the total variation distance.  
Hence, $0 \le g_{\theta, \theta^{\prime},\sigma}(X,Y) \le \|g_{\theta, \theta^{\prime}, \sigma}\|_{\infty}$ with probability one and this is equivalent with 
$$ \left|g_{\theta, \theta^{\prime},\sigma}(X,Y) - \frac{\|g_{\theta, \theta^{\prime},\sigma}\|_{\infty}}{2}\right| \le \frac{\|g_{\theta, \theta^{\prime}, \sigma}\|_{\infty}}{2}$$ with probability one. 
We recall the well-known variational characterization of the variance. 
Therefore, 
\[ \textup{Var}\left(g_{\theta, \theta^{\prime},\sigma}(X,Y)\right) = \min_{t \in \mathbb{R}} E\left[(g_{\theta, \theta^{\prime},\sigma}(X,Y)-t)^2 \right] \le E\left[\left(g_{\theta, \theta^{\prime},\sigma}(X,Y)- \frac{\|g_{\theta, \theta^{\prime},\sigma}\|_{\infty}}{2}\right)^2 \right] \le  \frac{\|g_{\theta, \theta^{\prime},\sigma}\|_{\infty}^2}{4}. \]
Now assertion (ii) follows from this estimate.

(iii) This follows from Eq.~\eqref{eq:fdiv-transformation} for $f(x) = \frac{1}{2}|x-1|$, assertion (ii), the Chebyshev-Cantelli inequality (\cite[Exercise 2.3]{Boucheron-2013}) and the Azuma-Hoeffding inequality (\cite[Theorem 2.8]{Boucheron-2013}).  
\end{proof}

The value of $\|g_{\theta, \theta^{\prime},\sigma}\|$ depends on $\sigma$, however we do not consider an optimal $\sigma$ here.

\subsubsection{Numerical computations}

Here we give numerical computations for $D_{f}[p_{\theta}: p_{\theta^{\prime}}]$ with $f(x) = \frac{1}{2}|x-1|$, that is, the total variation distance.  
We consider the following cases: 
\begin{align*}
    (\theta,\theta^{\prime}) \in & \left\{ ((1,0,0),(2,1,1)), ((1,0,0),(3,1,1)), ((1,0,0),(4,1,1)), 
((1,0,0),(4,3,2)),\right. \\
& \left. ((2,1,1),(3,1,1)), ((2,1,1),(4,1,1)), ((3,1,1),(4,1,1)), ((4,1,1),(4,3,2))\right\}.
\end{align*} 

We use the statistical software Python for the random number generation following the hyperboloid distribution. the statistical software R for the Monte-Carlo methods. 
 
The RNG means the direct approximation by the random number generation following the hyperboloid distribution. 
It depends on the approximation 
\[ E[g_{\theta, \theta^{\prime}}(X,Y)] \approx \frac{1}{2n} \sum_{i=1}^{n} \left|\frac{p_{\theta^{\prime}}(X_i,Y_i)}{p_{\theta}(X_i,Y_i)}- 1\right|, \]
where $(X_i, Y_i)$ are independent and identically distributed and $(X_1, Y_1)$ follows the hyperboloid distribution with parameter $\theta$. 

In (MC1), we optimize $\sigma > 0$ in 
\[ \frac{1}{n} \sum_{i=1}^{n} \frac{H_{\theta,\theta^{\prime}}(X_i,Y_i)^2}{p_{\sigma}(X_i)p_{\sigma}(Y_i)p(X_i)p(Y_i)} \]
by using the R package {\tt optimize}. 
(MC1i) deals with the case that $p(x)$ is the logistic distribution and (MC1ii) deals with the case that $p(x)$ is the standard $t$-distribution with degree of freedom $7$. 
We let $n = 10^8$ in each model.

\begin{table}[H]
\centering
\begin{tabular}{|c||c||c|c|c|}\hline
$(\theta,\theta^{\prime})$ & RNG & (MC1i)  &  (MC1ii) &  (MC2)  \\ \hline
 ((1,0,0),(2,1,1))   & 0.4667749 & 0.4684961 & 0.468601 & 0.4684339 \\
  ((1,0,0),(3,1,1))   & 0.4431547 & 0.4310651 & 0.4310781 & 0.4310919 \\
  ((1,0,0),(4,1,1))   & 0.4760025 & 0.4868136  & 0.4868225 & 0.4868233 \\
  ((1,0,0),(4,3,2))   & 0.6855790 & 0.7194658  & 0.7199457 & 0.7193469\\
  ((2,1,1),(3,1,1))   & 0.3125775 & 0.3131345  & 0.3132867 & 0.31312\\
((2,1,1),(4,1,1))     & 0.4486406 & 0.4543952 & 0.4546337 & 0.4544327\\
 ((2,1,1),(4,3,2))     & 0.3862757 & 0.3865376   & 0.3868286 & 0.3864432\\
  ((3,1,1),(4,1,1))     & 0.1636375 & 0.1635603  & 0.163615 & 0.1635727\\
   ((3,1,1),(4,3,2))    & 0.607509 & 0.6070837 & 0.6076672 & 0.607068\\
    ((4,1,1),(4,3,2))   & 0.7106694 & 0.7102112  & 0.7110066 & 0.7103308\\ \hline
\end{tabular}
\caption{numerical computations for $E\left[g_{\theta, \theta^{\prime}}(X,Y)\right]$}
\end{table}

The following are sample variances. 
They approximate $\textup{Var}\left(g_{\theta, \theta^{\prime}}(X,Y)\right)$. 
We do not deal with the case of the random number generation (RNG) by Lemma \ref{lem:var-infinite-rng}. 

\begin{table}[H]\label{table:sample_variance}
\centering
\begin{tabular}{|c||c|c|c|}\hline
$(\theta,\theta^{\prime})$ &  (MC1i)  &  (MC1ii) &  (MC2)  \\ \hline
 ((1,0,0),(2,1,1))    &   0.194733 & 0.1927373 & 1.228994 \\
  ((1,0,0),(3,1,1))   &   0.09103575 & 0.08650681 & 0.2690483 \\
  ((1,0,0),(4,1,1))   &   0.1698632  & 0.1591417 & 0.2302013 \\
  ((1,0,0),(4,3,2))   &   1.210232  & 1.163058  & 9.187897\\
  ((2,1,1),(3,1,1))   &   0.3186532  & 0.3176044 & 0.998442\\
((2,1,1),(4,1,1))     &    0.695998 & 0.6837022 & 1.395525\\
 ((2,1,1),(4,3,2))     &   0.5976724   & 0.5847235 & 5.240727\\
  ((3,1,1),(4,1,1))     &  0.0445869  & 0.04507904 & 0.08575983 \\
   ((3,1,1),(4,3,2))    &   1.917317 &  1.821023 & 8.963698\\
    ((4,1,1),(4,3,2))   &   3.101666  & 2.875278 & 9.606043\\ \hline
\end{tabular}
\caption{numerical computations for $\textup{Var}\left(g_{\theta, \theta^{\prime}}(X,Y)\right)$}
\end{table}

See Sections \ref{sec:R} and  \ref{sec:RNG} for the detailed source codes and other information such as the repetition numbers of random generations.  

\section{Conclusion}\label{sec:concl}

In this paper, we have considered two exponential families of hyperbolic distributions and studied various information-theoretic measures and their underlying information geometry.
The first hyperbolic exponential family is defined on the sample space of the Poincar\'e upper plane model and be either parameterized by 3D vector of equivalent $2\times 2$-symmetric positive-definite matrix (\S\ref{sec:PoincareVM}).  
We proved that all $f$-divergences between Poincar\'e distributions can be expressed using three canonical terms (\S\ref{sec:fdiv}), and reported  their Kullback-Leibler divergence (\S\ref{sec:kld}), their differential entropy (\S\ref{sec:h}). 
The second hyperbolic exponential family is defined in on the sample space of the hyperboloid model in arbitrary dimension.
We prove that statistical mixtures of hyperboloid distributions are universal density approximators of smooth densities (\S\ref{sec:mixuda}), and exhibited a correspondence between hyperboloid and Poincar\'e distributions when the sample space is 2D (\S\ref{sec:correspondence}).
Finally, we described two Monte Carlo methods to estimate numerically $f$-divergences between hyperboloid distributions in 
\S\ref{sec:numerics}.
We expect that these hyperbolic distributions and their mixtures will prove important in machine learning increasingly dealing with hierarchical structured datasets  embedded in hyperbolic spaces~\cite{cho2022rotated,song2022preliminary}.

\bibliographystyle{plain}
\bibliography{PoincareHyperboloidBIB.bib}

\appendix
\section{Calculations with the computer algebra system {\tt Maxima}}\label{app:Maxima} 

The Fisher information metric is expressed using the natural coordinate system as the Hessian of the log-normalizer.
Using {\tt Maxima} (\url{https://maxima.sourceforge.io/}), we can calculate symbolically the Hessian using  the following  snippet code:

\begin{lstlisting}[caption={Calculate symbolically Fisher metric and the components of the cubic tensor for the Poincar\'e distributions},captionpos=t,frame=single,xleftmargin=4pt,xrightmargin=3.5pt,linewidth=1\linewidth,
basicstyle=\footnotesize, 
commentstyle=\footnotesize,
keywordstyle=\color{blue},
breaklines=true]
F(a ,b, c):=log ( %pi/((sqrt (a*c-b*b) *exp(2* sqrt (a*c-b*b) ) )) ) ;
hessian(F(a,b,c),[a,b,c]);
tex(ratsimp(%));
/* Calculate T_{123} */
derivative(F(a,b,c),a,1);derivative(%,b,1);derivative(%,c,1);
tex(ratsimp(%));
\end{lstlisting}

\section{Calculations with the statistical software {\tt R}}\label{sec:R} 

\begin{lstlisting}[caption={Calculate in R the KLD and TVD for between hyperboloid distributions},captionpos=t,frame=single,xleftmargin=4pt,xrightmargin=3.5pt,linewidth=1\linewidth,
basicstyle=\footnotesize, 
commentstyle=\footnotesize,
keywordstyle=\color{blue},
breaklines=true]
basedensity <- function(x,z,w)
(sqrt(x[1]^2-x[2]^2-x[3]^2)*exp(sqrt(x[1]^2-x[2]^2-x[3]^2))/(2*pi))
*(exp(-x[1]*sqrt(1+z^2+w^2) +x[2]*z + x[3]*w)/sqrt(1+z^2+w^2))

H <- function(a,b,z,w) abs(basedensity(a,z,w)-basedensity(b,z,w))/2

ELOGIS <- function(s,a,b,z,w) 
H(a,b,z,w)^2*(1+exp(z))^2*exp(-z)*(1+exp(w))^2*exp(-w)
*(1+exp(z/s))^2*exp(-z/s)*(1+exp(w/s))^2*exp(-w/s)*s^2 #(1i)

plogiss <- function(s,x,z,w) 
basedensity(x,z,w)
*(1+exp(z/s))^2*exp(-z/s)*(1+exp(w/s))^2*exp(-w/s)*s^2 #(1i)

ET7 <- function(s,a,b,z,w) 
H(a,b,z,w)^2*(s^2*25*pi^2*((z/s)^2+7)^4*((w/s)^2+7)^4/(7*5488^2))
*(25*pi^2*((z)^2+7)^4*((w)^2+7)^4/(7*5488^2)) #(1ii)

pt7s <- function(s,x,z,w) 
basedensity(x,z,w)*(s^2*25*pi^2*((z/s)^2+7)^4*((w/s)^2+7)^4/(7*5488^2))  #(1ii)

bdpolar <- function(x,r,s) 
basedensity(x,r*cos(s)/sqrt(1-r^2),r*sin(s)/sqrt(1-r^2))*2*pi*r/(1-r^2)^2 #(2) 

TVD <- function(x,y) mean(abs(x-y)/2)  #total variation
VARTVD <- function(x,y) var(abs(x-y)/2) #sample variance of TVD

#The repetition number is $10^8$. 

#(MC1i)
z <- rlogis(10^8);w <- rlogis(10^8) #random number generation

ELOGIS01 <- function(s) mean(ELOGIS(s,c(1,0,0),c(2,1,1),z,w))

optimize(ELOGIS01,c(-10^4,10^4))$minimum
[1] 1.346247

x0 <- plogiss(1.346247,c(1,0,0),1.346247*z,1.346247*w)
x1 <- plogiss(1.346247,c(2,1,1),1.346247*z,1.346247*w)

TVD(x0,x1)
[1] 0.4684961
VARTVD(x0,x1)
[1] 0.194733
     
#(MC1ii) 
z <- rt(10^8,df=7);w <- rt(10^8,df=7) #random number generation

ET701 <- function(s) mean(ET7(s,c(1,0,0),c(2,1,1),z,w))

optimize(ET701,c(-10^4,10^4))$minimum
[1] 2.127577

y0 <- pt7s(2.127577,c(1,0,0),2.127577*z,2.127577*w)
y1 <- pt7s(2.127577,c(2,1,1),2.127577*z,2.127577*w)

TVD(y0,y1)
[1] 0.468601
VARTVD(y0,y1)
[1] 0.1927373

#(MC2) 
r <- runif(10^8,0,1-10^{-4});s <- runif(10^8,0,2*pi) #random number generation
z0 <- bdpolar(c(1,0,0),r,s);z1 <- bdpolar(c(2,1,1),r,s)
TVD(z0,z1)
[1] 0.4684339
VARTVD(z0,z1) 
[1] 1.228994

\end{lstlisting}

The computation for (2) does not work well if we let 
\texttt{r <- runif(10\^\ 7,0,1)}.
The reason will be the fact that the function $r \mapsto r/(1-r^2)^2$ has bad integrability around $r = 1$. 

\section{Random number generation of hyperboloid variates in {\tt Python}}\label{sec:RNG}

\begin{lstlisting}[caption={Random variate generation of a hyperboloid distribution},captionpos=t,frame=single,xleftmargin=4pt,xrightmargin=3.5pt,linewidth=1\linewidth,
basicstyle=\footnotesize, 
commentstyle=\footnotesize,
keywordstyle=\color{blue},
breaklines=true]
import random
import statistics
import math
import numpy
from scipy.stats import geninvgauss

n = 100000000
def mink(a,b):
    return a[0]*b[0]-a[1]*b[1]-a[2]*b[2]
def hyperboloid2Ddensity(a,x1,x2):
    return (1/(2*math.pi))*numpy.sqrt(mink(a,a))*numpy.exp(numpy.sqrt(mink(a,a))) 
    *numpy.exp(-(a[0]*numpy.sqrt(1+x1*x1+x2*x2)-a[1]*x1-a[2]*x2))/numpy.sqrt(1+x1*x1+x2*x2)
a = [2,1,1] 
b = [1,0,0]
q, c = 0.5, math.sqrt(mink(a,a))
s = geninvgauss.rvs(q,c, size=n)/c
ss = numpy.sqrt(s)
x1 = numpy.random.normal(loc=s*a[1], scale=ss,size=n)
x2 = numpy.random.normal(loc=s*a[2], scale=ss,size=n)
d1 = numpy.array(hyperboloid2Ddensity(a,x1,x2))
d2 = numpy.array(hyperboloid2Ddensity(b,x1,x2))
d3 = abs(d2/d1 - 1)/2

TVD = statistics.mean(d3)
print(TVD)
\end{lstlisting}

\section{Conversions between main models of hyperbolic geometry}\label{sec:models}

A probability density function $p_{M_1}(x,y)$ defined in one model $M_1$ of hyperbolic geometry can be transferred into another probability density function $p_{M_2}(x',y')$ of another model $M_2$ of hyperbolic geometry.
Figure~\ref{fig:convmodels} displays the conversions between the six main models of hyperbolic geometry.

\begin{figure}[H]
\centering
\includegraphics[width=0.7\columnwidth]{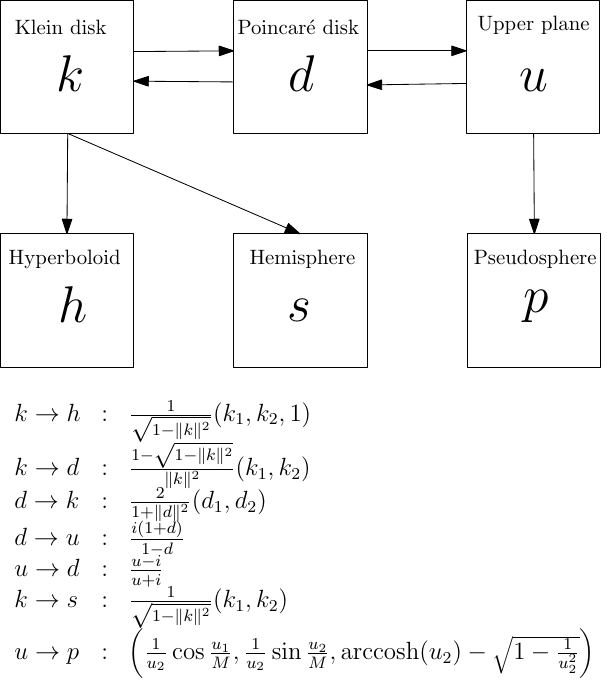}%
\caption{Conversion procedures between six main models of hyperbolic geometry.}%
\label{fig:convmodels}%
\end{figure}
\def\Jac{\mathrm{Jac}}

Let $(x',y')=m_{12}(x,y)$ denote the smooth differentiable mapping to convert $(x,y)$ of $M_1$ to $(x',y')$ of $M_2$ 
(with $(x,y)=m_{21}(x,y)=m_{12}^{-1}(x,y)$). 
We have
$$
p_{M_2}(x',y')=\Jac_{m_{12}^{-1}}(x',y') p_{M_1}(x,y),
$$
with 
$$
\Jac_{m_{12}^{-1}}(x',y')=\left[\frac{\partial m_{12}^{-1}(x,y)}{\partial(x,y)}\right].
$$

Furthermore, since the $f$-divergences between two pds are invariant by smooth diffeomorphisms~\cite{qiao2010study} of their  sample space, we have for any two densities $p_{M_1}$ and $p_{M_1}'$ with corresponding densities $p_{M_2}$ and $p_{M_2}'$ on the model $M_2$ :
$$
D_f[p_{M_1}:p_{M_1}']=D_f[p_{M_2}:p_{M_2}'].
$$


\section{Fisher-Rao geometry of the hyperboloid distribution}\label{app:FIM}

The Fisher Information Matrix (FIM) is covariant under reparameterization.
In \cite[p.126]{barndorff1989decomposition} , the FIM for a different local coordinate is reported. 
We have the following diffeomorphism: 
\[ \Theta \simeq (0,\infty) \times \mathbb{R}^2 \]
\[ \theta = (\theta_0, \theta_1, \theta_2) \mapsto (|\theta|, \theta_1/|\theta|, \theta_2/|\theta|)= (y_1,y_2,y_3)  \]

For $y = (y_1,y_2,y_3) \in (0,\infty) \times \mathbb{R}^2$, 
\[ 
(g_{i j}(y))_{i,j} 
= 
\begin{pmatrix}
y_1^{-2} & 0  & 0 \\
0 & \frac{(1+y_1)(1+y_3^2)}{1 + y_2^2 + y_3^2} & - \frac{(1+y_1)y_2 y_3}{1 + y_2^2 + y_3^2}  \\
0 & - \frac{(1+y_1)y_2 y_3}{1 + y_2^2 + y_3^2}  & \frac{(1+y_1)(1+y_2^2)}{1 + y_2^2 + y_3^2}
\end{pmatrix}.
\]
This expression is different from \eqref{eq:fim-hyp-2D}. 

The inverse matrix is given by 
\[ 
(g^{i j}(y))_{i,j} 
= 
\begin{pmatrix}
y_1^{2} & 0  & 0 \\
0 & \frac{1+y_2^2}{1 + y_1} & \frac{y_2 y_3}{1 + y_1}  \\
0 &  \frac{y_2 y_3}{1 + y_1}  & \frac{1+y_3^2}{1 + y_1}
\end{pmatrix}
\]

We use Wolfram Mathematica\textregistered{} for computations. 
We use functions given by \cite{Fairchild2020}. 
\begin{verbatim}
g = {{1/y1^2, 0, 0}, 
     {0, (1 + y1)*(1 + y3^2)/(1 + y2^2 + y3^2), -(1 + y1)*y2*y3/(1 + y2^2 + y3^2)}, 
     {0, -(1 + y1)*y2*y3/(1 + y2^2 + y3^2), (1 + y1)*(1 + y2^2)/(1 + y2^2 + y3^2)}};
xx = {y1, y2, y3};
\end{verbatim}

\subsection{Christoffel symbol of the second kind}

We use the following function: 
\begin{verbatim}
Christoffel[coordinates_,MetricTensor_]:=Module[ {n = Length[coordinates]},
Simplify[ Inverse[MetricTensor].(1/2Table[ D[MetricTensor[[s, j]], coordinates[[k]]] +
D[MetricTensor[[s, k]], coordinates[[j]]] -
D[MetricTensor[[j, k]], coordinates[[s]]], {s, n}, {j, n}, {k, n}]) ] ]
\end{verbatim}

The input is 
\begin{verbatim}
Christoffel[xx,g]
\end{verbatim}

The output is 
\begin{verbatim}
{{{-(1/y1), 0, 0}, 
{0, -((y1^2 (1 + y3^2))/(2 (1 + y2^2 + y3^2))), (y1^2 y2 y3)/(2 (1 + y2^2 + y3^2))}, 
{0, (y1^2 y2 y3)/(2 (1 + y2^2 + y3^2)), -((y1^2 (1 + y2^2))/(2 (1 + y2^2 + y3^2)))}}, 
{{0, 1/(2 + 2 y1), 0}, 
{1/(2 + 2 y1), -((y2 (1 + y3^2))/(1 + y2^2 + y3^2)), (y2^2 y3)/(1 + y2^2 + y3^2)}, 
{0, (y2^2 y3)/(1 + y2^2 + y3^2), -((y2 + y2^3)/(1 + y2^2 + y3^2))}}, 
{{0, 0, 1/(2 + 2 y1)}, 
{0, -((y3 + y3^3)/(1 + y2^2 + y3^2)), (y2 y3^2)/(1 + y2^2 + y3^2)}, 
{1/(2 + 2 y1), (y2 y3^2)/(1 + y2^2 + y3^2), -(((1 + y2^2) y3)/(1 + y2^2 + y3^2))}}}
\end{verbatim}

\subsection{Riemannian curvature tensor}

\begin{verbatim}
RiemannContravariant[coordinates_, MetricTensor_] := 
Module[{n, c}, n = Length[coordinates]; c = Christoffel[coordinates, MetricTensor]; 
Simplify@Table[ D[c[[i, j, l]], coordinates[[k]]] - D[c[[i, j, k]], coordinates[[l]]] 
+ (c . c)[[i, k, l, j]] - (c . c)[[i, l, k, j]], {i, n}, {j, n}, {k, n}, {l, n}]]
\end{verbatim}

The input is 
\begin{verbatim}
Contravariant[xx,g]
\end{verbatim}

The output is 
{\scriptsize
\begin{verbatim}
{{{{0, 0, 0}, {0, 0, 0}, {0, 0, 0}}, 
{{0, -((y1 (2 + y1) (1 + y3^2))/(4 (1 + y1) (1 + y2^2 + y3^2))), (y1 (2 + y1) y2 y3)/(4 (1 + y1) (1 + y2^2 + y3^2))}, 
{(y1 (2 + y1) (1 + y3^2))/(4 (1 + y1) (1 + y2^2 + y3^2)), 0, 0}, 
{-((y1 (2 + y1) y2 y3)/(4 (1 + y1) (1 + y2^2 + y3^2))), 0, 0}}, 
{{0, (y1 (2 + y1) y2 y3)/(4 (1 + y1) (1 + y2^2 + y3^2)), -((y1 (2 + y1) (1 + y2^2))/(4 (1 + y1) (1 + y2^2 + y3^2)))}, 
{-((y1 (2 + y1) y2 y3)/(4 (1 + y1) (1 + y2^2 + y3^2))), 0, 0}, 
{(y1 (2 + y1) (1 + y2^2))/(4 (1 + y1) (1 + y2^2 + y3^2)), 0, 0}}}, 
{{{0, (2 + y1)/(4 y1 (1 + y1)^2), 0}, {-((2 + y1)/(4 y1 (1 + y1)^2)), 0, 0}, {0, 0, 0}}, 
{{0, 0, 0}, {0, 0, ((2 + y1)^2 y2 y3)/(4 (1 + y1) (1 + y2^2 + y3^2))}, 
{0, -(((2 + y1)^2 y2 y3)/(4 (1 + y1) (1 + y2^2 + y3^2))), 0}}, 
{{0, 0, 0}, {0, 0, -(((2 + y1)^2 (1 + y2^2))/(4 (1 + y1) (1 + y2^2 + y3^2)))}, 
{0, ((2 + y1)^2 (1 + y2^2))/(4 (1 + y1) (1 + y2^2 + y3^2)), 0}}}, 
{{{0, 0, (2 + y1)/(4 y1 (1 + y1)^2)}, {0, 0, 0}, {-((2 + y1)/(4 y1 (1 + y1)^2)), 0, 0}}, 
{{0, 0, 0}, {0, 0, ((2 + y1)^2 (1 + y3^2))/(4 (1 + y1) (1 + y2^2 + y3^2))}, 
{0, -(((2 + y1)^2 (1 + y3^2))/(4 (1 + y1) (1 + y2^2 + y3^2))), 0}}, 
{{0, 0, 0}, {0, 0, -(((2 + y1)^2 y2 y3)/(4 (1 + y1) (1 + y2^2 + y3^2)))}, 
{0, ((2 + y1)^2 y2 y3)/(4 (1 + y1) (1 + y2^2 + y3^2)), 0}}}}
\end{verbatim}
}

\subsection{Sectional curvature}

\begin{verbatim}
RiemannCovariant[coordinates_, MetricTensor_] := 
 MetricTensor . RiemannContravariant[coordinates, MetricTensor]
\end{verbatim}

The input is 
\begin{verbatim}
Simplify[RiemannCovariant[xx, g]]
\end{verbatim}

{\scriptsize
\begin{verbatim}
{{{{0, 0, 0}, {0, 0, 0}, {0, 0, 0}}, 
{{0, -(((2 + y1) (1 + y3^2))/(4 y1 (1 + y1) (1 + y2^2 + y3^2))), ((2 + y1) y2 y3)/(4 y1 (1 + y1) (1 + y2^2 + y3^2))}, 
{((2 + y1) (1 + y3^2))/(4 y1 (1 + y1) (1 + y2^2 + y3^2)), 0, 0}, 
{-(((2 + y1) y2 y3)/(4 y1 (1 + y1) (1 + y2^2 + y3^2))), 0, 0}}, 
{{0, ((2 + y1) y2 y3)/(4 y1 (1 + y1) (1 + y2^2 + y3^2)), -(((2 + y1) (1 + y2^2))/(4 y1 (1 + y1) (1 + y2^2 + y3^2)))}, 
{-(((2 + y1) y2 y3)/(4 y1 (1 + y1) (1 + y2^2 + y3^2))), 0, 0}, 
{((2 + y1) (1 + y2^2))/(4 y1 (1 + y1) (1 + y2^2 + y3^2)), 0, 0}}}, 
{{{0, ((2 + y1) (1 + y3^2))/(4 y1 (1 + y1) (1 + y2^2 + y3^2)), -(((2 + y1) y2 y3)/(4 y1 (1 + y1) (1 + y2^2 + y3^2)))}, 
{-(((2 + y1) (1 + y3^2))/(4 y1 (1 + y1) (1 + y2^2 + y3^2))), 0, 0}, 
{((2 + y1) y2 y3)/(4 y1 (1 + y1) (1 + y2^2 + y3^2)), 0, 0}}, 
{{0, 0, 0}, {0, 0, 0}, {0, 0, 0}}, 
{{0, 0, 0}, {0, 0, -((2 + y1)^2/(4 (1 + y2^2 + y3^2)))}, {0, (2 + y1)^2/(4 (1 + y2^2 + y3^2)), 0}}}, 
{{{0, -(((2 + y1) y2 y3)/(4 y1 (1 + y1) (1 + y2^2 + y3^2))), ((2 + y1) (1 + y2^2))/(4 y1 (1 + y1) (1 + y2^2 + y3^2))}, 
{((2 + y1) y2 y3)/(4 y1 (1 + y1) (1 + y2^2 + y3^2)), 0, 0}, 
{-(((2 + y1) (1 + y2^2))/(4 y1 (1 + y1) (1 + y2^2 + y3^2))), 0, 0}}, 
{{0, 0, 0}, {0, 0, (2 + y1)^2/(4 (1 + y2^2 + y3^2))}, {0, -((2 + y1)^2/(4 (1 + y2^2 + y3^2))), 0}}, 
{{0, 0, 0}, {0, 0, 0}, {0, 0, 0}}}}
\end{verbatim}
}

The input and output are 
\begin{verbatim}
RiemannCovariant[xx, g][[1, 3, 1, 3]]/(g[[3, 3]]*g[[1, 1]] - g[[1, 3]]^2)

-((y1 (2 + y1))/(4 (1 + y1)^2))

RiemannCovariant[xx, g][[1, 2, 1, 2]]/(g[[2, 2]]*g[[1, 1]] - g[[1, 2]]^2)

-((y1 (2 + y1))/(4 (1 + y1)^2))

RiemannCovariant[xx, g][[2, 3, 2, 3]]/(g[[3, 3]]*g[[2, 2]] - g[[2, 3]]^2)

-((2 + y1)^2/(4 (1 + y1)^2))
\end{verbatim}

\subsection{Ricci curvature}

\begin{verbatim}
RicciCurvature[coordinates_, MetricTensor_] := 
 Simplify@TensorContract[RiemannContravariant[coordinates, MetricTensor], {{1, 3}}]
\end{verbatim}

The input is 
\begin{verbatim}
RicciCurvature[xx, g]
\end{verbatim}

The output is 
{\small
\begin{verbatim}
{{-((2 + y1)/(2 y1 (1 + y1)^2)), 0, 0}, 
{0, -(((2 + y1) (1 + y3^2))/(2 (1 + y2^2 + y3^2))), ((2 + y1) y2 y3)/(2 (1 + y2^2 + y3^2))}, 
{0, ((2 + y1) y2 y3)/(2 (1 + y2^2 + y3^2)), -(((2 + y1) (1 + y2^2))/(2 (1 + y2^2 + y3^2)))}}
\end{verbatim}

}

It is easy to see the manifold is complete and simply-connected. 
Thus we see that 
\begin{proposition}\label{prop:Hadamard-not-Einstein}
The statistical manifold of the
$2$-dimensional hyperboloid distribution is Hadamard but not Einstein. 
\end{proposition}

Fisher-Rao manifolds which are Hadamard are useful to prove convergence of optimization algorithms \cite{bacak2014convex}.

\end{document}